\documentclass[runningheads]{llncs}
\usepackage[T1]{fontenc}
\usepackage{graphicx}
\usepackage{cite}
\usepackage{amsmath,amssymb,amsfonts}
\usepackage{algorithmic}
\usepackage{graphicx}
\usepackage{textcomp}
\usepackage{xcolor}
\def\BibTeX{{\rm B\kern-.05em{\sc i\kern-.025em b}\kern-.08em
    T\kern-.1667em\lower.7ex\hbox{E}\kern-.125emX}}

\usepackage{stmaryrd}
\usepackage{centernot}
\usepackage{enumerate}
\usepackage{url}
\usepackage{tikz}
\usepackage{pgf-umlsd}
\usepackage[linktoc=all,pdftex,pdfpagelabels,bookmarks,hyperindex,hyperfigures,hidelinks]{hyperref}

\input xy 
\xyoption{all}

\usepackage{comment}

\graphicspath{{figures/}}
\usepackage{caption}
\usepackage{subcaption}

\usepackage{relsize}
\usepackage{fancyvrb}
\fvset{numbersep=2pt, fontfamily=tt, fontsize=\smaller, frame=single}

\newcommand{\proto}{\caP}
\newcommand{\dlct}{\mathcal{D}}
\newcommand{\lingo}{\Lambda}

\usepackage{marginnote}
\usepackage[normalem]{ulem}
\newcommand{\editornote}[3]{%
    {\color{#3}%
        {\marginnote%
            [\color{#3}#2$\vartriangleright$]{\color{#3}$\vartriangleleft$#2}
        }%
        {#1}%
    }%
}
\newcommand{\victor}[1]{\editornote{#1}{Víctor}{blue}}
\newcommand{\jose}[1]{\editornote{#1}{Jose}{teal}}

\newcommand{\catherine}[1]{\editornote{#1}{Catherine}{red}}

\makeatletter
\newcommand*{\vneq}{%
  \mathrel{%
    \mathpalette\@vneq{=}%
  }%
}
\newcommand*{\@vneq}[2]{%
  \sbox0{\raisebox{\depth}{$#1\neq$}}%
  \sbox2{\raisebox{\depth}{$#1|\m@th$}}%
  \ifdim\ht2>\ht0 %
    \sbox2{\resizebox{\vneqxscale\width}{\vneqyscale\ht0}{\unhbox2}}%
  \fi
  \sbox2{$\m@th#1\vcenter{\copy2}$}%
  \ooalign{%
    \hfil\phantom{\copy2}\hfil\cr
    \hfil$#1#2\m@th$\hfil\cr
    \hfil\copy2\hfil\cr
  }%
}
\newcommand*{\vneqxscale}{1}
\newcommand*{\vneqyscale}{1}
\makeatother

\newcommand{\caP}{{\cal P}}

\begin{document}
\title{A Theory of Composable Lingos \\ for Protocol Dialects}
%
%

\author{V\'{\i}ctor Garc\'{\i}a\inst{1}
\and
Santiago Escobar\inst{1}
\and
Catherine Meadows\inst{2} 
\and
Jose Meseguer\inst{3}
}
\authorrunning{V. García et al.}
%
\institute{
VRAIN, Universitat Polit{\`e}cnica de Val\`{e}ncia, Valencia, Spain\\
\email{\{vicgarval,sescobar\}@upv.es} 
\and
\email{cameadows3306@gmail.com}
\and
University of Illinois at Urbana-Champaign, Illinois, USA\\
\email{meseguer@illinois.edu}
}

\maketitle              
\begin{abstract}
Formal patterns are  formally specified solutions
to frequently occurring distributed system problems that are generic, executable, and come
with strong qualitative and/or quantitative formal guarantees.  A formal pattern is a \emph{generic system transformation}
which transforms a usually infinite class of systems in need of the pattern's solution into enhanced versions
of such systems that solve the problem in question.
 In this paper we demonstrate the application of formal patterns to \emph{protocol dialects}. 
 Dialects are methods for hardening protocols so as to endow them with light-weight security, especially against easy attacks 
 that can lead to more serious ones.  A \emph{lingo} 
is a dialect's key security component, because attackers are unable to ``speak'' the lingo.  A lingo's ``talk'' changes all the time, becoming
a moving target for attackers.  In this paper we present several formal patterns for both lingos and dialects.
Lingo formal patterns can make lingos stronger by both \emph{transforming} them and by \emph{composing} several lingos
into a stronger lingo.  Dialects themselves can be obtained by the application of a single
\emph{dialect formal pattern}, generic on both the chosen lingo and the chosen protocol.

\end{abstract}
%
%
%
%
%

\section{Introduction}\label{sec:introduction}
Formal patterns \cite{MeseguerSCP14} are methods for formally specifying generic solutions to distributed system problems that are executable and provide formal guarantees. Such patterns allow one to retain what has been learned about solving such problems without tying that knowledge to a particular concrete solution. They also allow for the composition of solutions and  the composition of guarantees.

Formal patterns have been successfully applied to a number of different distributed systems problems \cite{duran2024programming,DBLP:conf/adt/Meseguer22,sun2010formal}.  
They have also been applied to security protocols, and 
have  been helpful in developing methods for hardening protocols against malicious attacks \cite{chadha2008modular,eckhardt2012stable,nigam2022automated}. 
In this paper we apply this methodology to an emerging  protocol hardening technique, the  design of \emph{protocol dialects}.
Protocol dialects are methods for modifying protocols in order to provide a light-weight layer of security, especially against relatively easy attacks that could potentially be leveraged into more serious ones. The scenario is usually that of a network of mutually trusting principals, e.g., an enterprise network, that needs to defend itself from outside attackers. 
The most effective approach is to prevent outside parties from even  initiating a communication with group members, e.g., by requiring messages to be modified in some way unpredictable by the attacker.\footnote{By ``message'' we  mean a unit of information transferred by the protocol, and can mean anything from a packet to a message in the traditional sense. }

Protocol dialects are intended to be used as a first line of defense, not as a replacement for traditional authentication protocols.  Indeed, they are primarily concerned with protecting communication between members of an enclave in which all parties trust each other against weak attackers that are trying to leverage off protocol vulnerabilities.  We can make use of this restricted scope to use an attacker model based on the on-path attacker model \cite{gogineniverify,ProtocolDialects-FormalPatterns}:  First, the attacker can read messages, but cannot destroy or replace them.  Second, principals  trust each other and share a common secret that they can use, for example, as the seed  for  a pseudo-random number generator to generate and share common  secret parameters.  
Because of the simplified threat model and trust assumption,
dialects can be designed to be simpler than full-scale authentication protocols.  Thus they can be expected to be prone to fewer implementation and configuration errors. 
With this in mind, dialects can be used as a separate but very thin layer in the protocol stack, useful as a defense against attacks in case the main authentication protocol is misconfigured or implemented incorrectly.   Considering the dialect as a separate layer makes it easier to keep it simple, because it makes it possible to avoid dependencies on the layers it is communicating with.  In our work, we consider dialects that sit right above the transport layer in the TCP/IP model, while the protocol it is modifying sits just above the dialect in the application layer.  However, other locations are possible for the protocol being modified. Variations in the threat model are also possible, e.g., deciding whether or not to include replay or session hijacking attacks, which can both be implemented by an on-path attacker.

Protocol dialects as formal patterns were first introduced in \cite{ProtocolDialects-FormalPatterns}. A dialect is generic on two parameters: the given \emph{lingo}, which describes the actual \emph{message transformation}, and the 
underlying \emph{protocol} whose security it improves. The
\emph{dialect} pattern is a \emph{protocol transformation} yielding a new protocol responsible for applying the lingo and interfacing with the underlying protocol. Lingos can be transformed and composed into stronger ones in
a compositional way.  Using lingo compositions a given protocol can be made
more secure in a wide variety of ways by means of the different dialects used in such lingo
compositions.

We substantially extend the work in \cite{ProtocolDialects-FormalPatterns} by concentrating on the design of lingos, instead of the design of dialects as a whole,
and developing \emph{new formal patterns} to create and transform new lingos from simpler ones.  This is a more flexible and computationally more efficient approach that 
increases the security of  transformed lingos (and therefore of dialects) and provides a
simpler methodology for designing new lingos and dialects out of previous ones and increasing their security.  In particular, we are able to \emph{replace} two important
formal patterns for dialects in \cite{ProtocolDialects-FormalPatterns} by much simpler similar formal patterns for lingos.
 Lingos rely on constantly changing secret \emph{parameters}, which are generated by a  pseudorandom number generator shared throughout, itself parametric on a single shared secret, an approach to lingo design introduced by  Gogeneni et al. in \cite{gogineniverify}.

 From the software engineering point of view, the importance of formal patterns ---in particular that of lingos
 and dialects--- is that they make possible the design and verification of distributed systems
 of much higher quality than those developed in a traditional way, and they endow such systems with desired new features
 ---for dialects with new lightweight security features--- in a fully modular way.  As explained in \cite{DBLP:conf/fase/Meseguer25},
 formal patterns
 provide a modular, generic way of developing formal executable specifications of system \emph{designs} that can be subjected to
systematic formal verification \emph{before} they are implemented and support the correct by construction development of high quality systems
with a notable economy of effort and a high degree of reusability.  Mathematically, the \emph{dialect formal pattern}
is a \emph{theory transformation} definable as a partial function:
\[
\mathcal{D}: \mathbf{PLingos} \times \mathbf{Protocols} \ni (\Lambda_{p},\mathcal{P}) \mapsto \mathcal{D}_{\Lambda_{p}}(\mathcal{P}) \in \mathbf{Protocols} 
\]
 where $\mathbf{PLingos}$ and $\mathbf{Protocols}$ denote classes of executable specifications, i.e., of executable
 \emph{theories}, respectively specifying the classes of parametric lingos and communication protocols.  Likewise,
 as we explain in the paper, \emph{lingo formal
 patterns} are similar theory transformations mapping one or more lingos to a new lingo enjoying specific
 formal properties.

\subsection{Related Work}\label{subsec:related-work}
Probably, the first work on dialects is that of Sjoholmsierchio et al. in \cite{sjoholmsierchio_software-defined_2019}.  In this work, the dialect is proposed as a variation of an open-source protocol to introduce new security measures or remove  unused code, and is applied to the OpenFlow protocol, introducing, among other things a defense against cipher-suite downgrade attacks on TLS 1.2.  This variation was independent of TLS, and was achieved by adding proxies, thus allowing for modification without touching the rest of the system and foreseeing the use of dialects as thin layers.  In later work by Lukaszewski and Xie \cite{lukaszewski2022towards} a layer 4.5 in the TCP/IP model was proposed for dialects.  In our own model we also make use of proxies for both sending and receiving parties, but stop short of proposing an additional official layer.

The threat model we use in this paper, and the one used in \cite{ProtocolDialects-FormalPatterns},
is similar to that used by Goginieni et al. in \cite{gogineniverify} and Ren et. al \cite{ren2023breaking}: an on-path attacker who is not able to corrupt any members of the enclave, as discussed earlier. 
We have also followed the suggestion in \cite{ren2023breaking} to view dialects as composable protocol transformations.  In this paper, we simplify this approach by reducing dialect transformations to more basic lingo transformations, which are generic, compositional methods for producing new lingos (and thus new dialects) from old ones.

In  \cite{garcia-alfaro_mpd_2021,gogineni_framework_2022,gogineniverify} Mei, Gogineni, et al. introduce an approach to dialects in which the message transformation is updated each time using a shared pseudo-random number generator.  This means that the security of a transformation depends as much, and possibly more, on its unpredictability as it does on the inability of the attacker to reproduce a particular  instance of a transformation.  We have adopted and extended this approach.  In particular, we model a transformation as a parametrized function, or \emph{lingo}, that produces a new message using the original message and the parameter as input.  Both the lingo and the parameter can be chosen pseudo-randomly, as described in  \cite{ProtocolDialects-FormalPatterns}. 

None of the work cited above (except \cite{ProtocolDialects-FormalPatterns})
applies formal design and evaluation techniques.   However, in \cite{talcott2024dialects} Talcott applies formal techniques to the study of dialects running over unreliable transport, such as UDP.   \cite{talcott2024dialects}  is complementary to our work in a number of ways.  First, we have been looking at dialects running over reliable transport, such as TCP.  Secondly, we concentrate on a particular attacker model, the on-path attacker with no ability to compromise keys or corrupt participants, and then explore ways of generating dialects and lingos that can be used to defend against this type of attacker.  In \cite{talcott2024dialects}, however, the dialects are relatively simple, but the attacker models are explored in more detail,  giving the attacker a set of specific actions that it can perform.   Like our model, the work in \cite{talcott2024dialects} is formalized in the Maude formal specification language, providing a clear potential synergy between the two models in the future.

 \subsection{Contributions of this paper}
 \begin{enumerate}
 \item We present a theory of lingos that  adds important new \emph{properties}, \emph{transformations} and \emph{composition operations} to
 the notion of lingo introduced in \cite{ProtocolDialects-FormalPatterns}.  The main goal of this theory is to improve the security of protocol dialects.  
 \item We show that these new lingo transformations and composition operations are \emph{formal patterns}  \cite{DBLP:conf/adt/Meseguer22} that guarantee
 desired properties and replace and greatly simplify 
 two earlier formal patterns for dialects in \cite{ProtocolDialects-FormalPatterns}.  
 These lingo transformations provide new compositional techniques for creating new  lingos and dialects.  
 \item We provide formal executable specifications in Maude \cite{maude-book} for various lingos, their transformations and compositions, and dialects, in Appendix \ref{app:examples}.  
 \end{enumerate}

 Section~\ref{sec:lingos} substantially extends lingos beyond \cite{ProtocolDialects-FormalPatterns}
 and provides a wealth of new results about lingos and lingo transformations.
 Section~\ref{sec:compositions} provides new lingo composition operations and proves their properties.
 Section~\ref{sec:dialects} recalls the dialect notion and illustrates it with various examples.
 Section~\ref{sec:conclusion} concludes the paper.

\section{Lingos}\label{sec:lingos}
A lingo is a data transformation $f$ between two data types
$D_{1}$ and $D_{2}$ with a one-sided inverse transformation $g$.  
The transformation $f$ is \emph{parametric} on a 
secret parameter value $a$ belonging to a parameter set $A$.  
For each parameter $a$, data from $D_{1}$ is transformed into data of $D_{2}$, which, using the same parameter $a$, can be transformed 
back into the original data from $D_{1}$ using the one-sided inverse $g$.
In all our applications,
lingos will be used to transform the payload of a message of sort $D_1$
appearing in some protocol $\mathcal{P}$
whereas $D_{2}$ will then be the data type of transformed payloads.
Such transformed payloads can then be sent, either as 
a single message or as a sequence of messages, to make
it hard for malicious attackers to interfere with 
the communication of honest participants, who are the
only ones knowing the current parameter $a \in A$.
The point of a lingo is that when such 
transformed messages are received by an honest
participant they can be transformed back using $g$ (if they were broken
into several packets they should first be reassembled) to get the original payload in $D_{1}$.

 Briefly, the set $D_1$ can be thought of as analogous to the plaintext space, the set $D_2$ as analogous to the ciphertext space in cryptographic systems, and the parameter set $A$ (also known as the secret parameter set) as analogous to the key space.
The elements of $A$ are generally required to be pseudorandomly generated from a secret shared by enclave members, so it can't be guessed in advance, and it is also generally updated with each use, so information gleaned from seeing one dialected message should not provide any help in breaking another.  However, it is not necessarily required that an attacker is not be able to guess $a$ after seeing $f(d_1,a)$.  The only requirement is that  it is not feasible for the attacker  to compute $f(d_1,a)$ without having been told that $a$ is the current parameter for
$f(d_1,a)$.  This is for two reasons.  First, the lingo is intended to provide authentication, not secrecy.  Secondly, the lingo is only intended to be secure against an on-path attacker \cite{ProtocolDialects-FormalPatterns} that can eavesdrop on traffic but not interfere with it. Thus, if an enclave member sends a message $f(d_1,a)$, the attacker, even if it learns $a$, can't remove $f(d_1,a)$ from the channel and replace it with $f(d'_1,a)$ where $d'_1$ is a message of the attacker's own choosing. It can send $f(d'_1,a)$ after $f(d_1,a)$ is sent, but $f(d'_1,a)$ may not be accepted if the secret parameter is changed each time the message is sent. The enforcement of the policy on the sharing and updating of parameters is the job of the dialect, which will be discussed in Section~\ref{sec:dialects}.


\begin{definition}[Lingo]\label{def:lingo}
A \emph{lingo} $\Lambda$ is a 5-tuple  $\Lambda = (D_1,D_2,A,f,g)$, where $D_1$, $D_2$ and $A$ are sets called, respectively, the \emph{input}, \emph{output}
and \emph{parameter} sets,
$f$, $g$ are functions $f : D_1 \times A \rightarrow D_2$,  $g : D_2 \times A \rightarrow D_1$ such that\footnote{The equality $g(f(d_1,a),a)=d_1$
can be generalized to 
an equivalence $g(f(d,a),a) \equiv d$. The generalization of lingos to allow for such message equivalence is left for future work.} 
$\forall d_1 \in D_1$, $\forall a \in A$, $g(f(d_1,a),a)=d_1$. 
We call a lingo \emph{non-trivial} iff $D_1$, $D_2$ and $A$
are non-empty sets. In what follows, all lingos considered will be
non-trivial.

Note that a lingo $\Lambda = (D_1,D_2,A,f,g)$ is just a three-sorted
$(\Sigma,E)$-algebra, with sorts\footnote{Note the slight abuse of notation:
in $\Sigma$, $D_{1}$, $D_2$ and $A$ are uninterpreted sort
\emph{names}, whereas in a given lingo
$\Lambda' = (D'_1,D'_2,A',f',g')$, such sorts are respectively interpreted
as \emph{sets} $D'_{1}$, $D'_2$ and $A'$ and, likewise, the uninterpreted
function \emph{symbols} $f$ and $g$ in $\Sigma$ are
interpreted as actual \emph{functions} $f'$ and $g'$ in a given lingo.}
$D_{1}$, $D_2$ and $A$, 
function symbols  $f : D_1 \times A \rightarrow D_2$ and $g : D_2 \times A \rightarrow D_1$, and $E$ the single $\Sigma$-equation $g(f(d_1,a),a)=d_1$.
%

\end{definition}

Note the asymmetry between $f$ and $g$, in the lingo
definition, since we
do not have an equation of the form 
$\forall d_2 \in D_2$, $\forall a \in A$, $f(g(d_2,a),a)=d_2$. 
The reason for this asymmetry is that, given $a \in A$,
the set $\{f(d_{1},a) \mid d_1 \in D_1 \}$ may be
a \emph{proper} subset of $D_{2}$.  However, we show
in Corollary \ref{cor:corollary2} below that
the equation $f(g(d_2,a),a)=d_2$ does hold for any
$d_2 \in \{f(d_{1},a) \mid d_1 \in D_1 \}$.
Other results and proofs are included in Appendix~\ref{app:proofs}.
Of course, for the special case of
a lingo such that  $\forall a \in A$,
$D_2 = \{f(d_{1},a) \mid d_1 \in D_1 \}$,
the equation 
$\forall d_2 \in D_2$, $\forall a \in A$, $f(g(d_2,a),a)=d_2$
will indeed hold. 
In particular, the set equality
$D_2 = \{f(d_{1},a) \mid d_1 \in D_1 \}$
holds for all $a \in A$ in the following example.

\begin{example}[The XOR$\{n\}$ Lingo] \label{ex:xor-lingo} 
Let $\Lambda_{\mathit{xor}}\{n\}=(\{0,1\}^{n},\{0,1\}^{n},\{0,1\}^{n},\oplus,\oplus)$, with $\_\oplus\_ : \{0,1\}^{n} \times \{0,1\}^{n} \rightarrow \{0,1\}^{n}$ the bitwise exclusive or operation. 

Note that $\Lambda_{\mathit{xor}}\{n\}$ is \emph{parametric} on $n \in \mathbb{N}\setminus \{0\}$.  That is, for each choice of $n \geq 1$ we get a corresponding lingo. This is a common phenomenon: for many lingos, $D_1, D_2$ and $A$ are not fixed sets, but \emph{parametrized data types}, so that for each choice of the parameter we get a corresponding instance lingo. For a more detailed explanation of \emph{parametrized data types}, see Appendix~\ref{app:param}.
\end{example}

\begin{example}[XOR-BSEQ Lingo for Bit-sequences] \label{ex:XOR-BSeq-Lingo}
Yet another variation on the same theme is to
define the $\Lambda_{xor.\mathit{BSeq}}$
lingo, whose elements are
\emph{bit sequences} of arbitrary length.  
\end{example}

Note that 
$\{0,1\}^{n}$ 
is not the only possible parametrized data type on which a
lingo based on the $\mathit{xor}$ operation could be based.  A different
parametrized data type can have as elements the \emph{finite subsets} of
a parameter set $D$, and where the $\mathit{xor}$
operation is interpreted as the \emph{symmetric difference}
of two finite subsets of $D$.  A finite subset, say, $\{a,b,c\} \subseteq D$,
with $a,b,c \in D$ different elements, can then be represented
as the expression $a \; \mathit{xor} \; b \; \mathit{xor} \; c$.
%
%
For $D$ a finite set, there is an isomorphism
between its power set $\mathcal{P}(D)$ with
symmetric difference, and the function set
$[D \rightarrow \{0,1\}]$ with pointwise $\mathit{xor}$
of predicates $p,q \in [D \rightarrow \{0,1\}]$, i.e., 
$p \; \mathit{xor} \; q = \lambda d \in D.\; p(d) \;\mathit{xor} \; q(d)$,
which when $|D|=n$
is isomorphic to $\{0,1\}^{n}$ with pointwise $\mathit{xor}$.

\begin{example}[Divide and Check (D\&C) Lingo] \label{ex:d&c-lingo-old}\label{ex:d&c-lingo}
Let $\Lambda_{D\&C} = (\mathbb{N},\mathbb{N} \times \mathbb{N},\mathbb{N},
f, g)$, given $\forall n, a, x, y \in \mathbb{N}$ then:
\begin{itemize}
 \item $f(n,a) = (quot(n +(a+2),a +2),rem(n + (a+2),a+2))$
\item $g((x,y),a) = (x \cdot (a+2)) + y -(a+2)$
\end{itemize}
where $\mathit{quot}$ and $\mathit{rem}$ denote the quotient
and remainder functions on naturals.
\end{example}

\noindent The idea of the $\Lambda_{D\&C}$ lingo is quite simple.  Given a parameter 
$a \in \mathbb{N}$,
an input number $n$ is transformed into a
pair of numbers $(x,y)$: the quotient $x$ of $n+a+2$ by $a+2$ (this makes sure that
$a+2 \geq 2$), and the remainder $y$ of $n+a+2$ by $a+2$.  The meaning of
``divide'' is obvious.  The meaning of ``check''
reflects the fact that a receiver of a pair $(x,y)$ who knows $a$
can check $x = quot((x \cdot (a+2))+y,a+2))$
and $y = rem((x \cdot (a+2))+y,a+2))$, giving some assurance
that the pair $(x,y)$ was obtained from
$n = (x \cdot (a+2)) + y -(a+2)$ and
has not been tampered with.  This
is an example of an $f$-\emph{checkable} lingo (see Section~\ref{f-check-subsection} below), whereas no such check is possible for any of the
isomorphic versions of the $\Lambda_{\mathit{xor}}$ lingo
(see again Section~\ref{f-check-subsection} below). 

For a detailed explanation with Maude formal executable specifications of Examples \ref{ex:xor-lingo}, \ref{ex:XOR-BSeq-Lingo} and \ref{ex:d&c-lingo}, we refer the reader to Appendix~\ref{app:ex-lingos}.

\subsection{f-Checkable Lingos and the $\Lambda \mapsto \Lambda^{\sharp}$ Lingo Transformation} \label{f-check-subsection}


In this section we discuss lingos where the validity of the generated payloads can be checked by the dialect. We define a lingo transformation that 
substantially reduces the probability of an attacker successfully forging a message.



\begin{definition}[f-Checkable Lingo]  A Lingo $\Lambda=(D_1,D_2,A,f,g)$ is called
$f$-\emph{checkable} iff $\forall a \in A$, $\exists d_2 \in D_2$
s.t. $\nexists \; d_1 \in D_1$ s.t. $d_2 = f(d_1,a)$.
\end{definition}

Above, $\Lambda$ is understood, not as a fixed lingo, but as an equational theory specifying 
lingos.  Therefore, the ``f'' in ``$f$-Checkable Lingo'' is
an uninterpreted function symbol.  However,
once an actual interpretation for the theory $\Lambda$
is given, and therefore for $D_1$, $D_2$, $A$, $f$ and $g$,
$f$ becomes interpreted as the  concrete function specified by that interpretation. 

Note that the $\Lambda_{\mathit{xor}}$  lingo of Example~\ref{ex:xor-lingo} is not $f$-checkable.
We call a lingo $\Lambda=(D_1,D_2,A,f,g)$
\emph{symmetric} iff $(D_2,D_1,A,g,f)$ is also a lingo.
Obviously, $\Lambda_{\mathit{xor}}$ is symmetric.  Any
symmetric lingo is \emph{not} $f$-checkable, since 
any $d_2 \in D_2$
satisfies the equation $f(g(d_2,a),a) = d_{2}$.

\begin{example} \label{ex:f-Ckeckable-lingo} The $\Lambda_{D \& C}$ lingo of Example~\ref{ex:d&c-lingo} is $f$-checkable.  Indeed,
for each $a \in \mathbb{N}$, any $(x,y)$ with $y \geq a +2$
cannot be of the form $(x,y) = f(n,a)$ for any $n \in \mathbb{N}$.
This is a cheap, easy check.  The actual check that a $d_1$ exists,
namely, the check $(x,y) = f(g((x,y),a),a)$ provided by Corollary \ref{cor:corollary3} in Appendix \ref{app:proofs} (where \emph{proofs of all theorems are given}),
was already explained when making explicit the meaning of "C" in 
$\Lambda_{D \& C}$.
\end{example}




The following theorem explains how to build an $f$-checkable lingo.
The proof is in Appendix~\ref{app:proofs}.

\begin{theorem}\label{theorem:trasnform-fCheckable}
    Let $\Lambda=(D_1,D_2,A,f,g)$ be a lingo, where
    we assume that $D_1$, $D_2$ and $A$ are computable data-types and $|D_{1}| \geq 2$.
    %
    %
    Then, $\Lambda^{\sharp}=(D_1,{D_2\times D_2},\linebreak {A \otimes A},f^{\sharp},g^{\sharp})$ is an
    $f$-checkable lingo, where:
    \begin{itemize}
        \item $A \otimes A = A \times A \setminus id_A$, with $id_A = \{(a,a) \in A^2 | a \in A\}$
        \item $f^{\sharp}(d_1, (a,a')) = (f(d_1,a), f(d_1,a'))$
        \item $g^{\sharp}((d_2,d'_2), (a,a')) = g(d_2,a)$.
        
    \end{itemize}
\end{theorem}

\begin{example}[$f$-checkable transformed version of 
$\Lambda_{\mathit{xor}}\{n\}$] \label{ex:edxor}
The lingo transformation $\Lambda \mapsto \Lambda^{\sharp}$
maps $\Lambda_{\mathit{xor}}\{n\}$ in Example \ref{ex:xor-lingo}
to the $f$-checkable lingo
$\Lambda_{\mathit{xor}}^{\sharp}\{n\}=(\{0,1\}^{n},\{0,1\}^{n}
\times \{0,1\}^{n},\{0,1\}^{n} \otimes \{0,1\}^{n},\oplus^{\sharp},
\oplus^{\sharp})$.  
\end{example}

\vspace{2ex}

%

\subsection{Malleable Lingos}\label{subsubsec:comp-malleable-lingos}

\noindent A cryptographic function is called \emph{malleable} 
if an intruder, \emph{without knowing the secret key},
can use one or more existing encrypted messages to generate another, related,
message also encrypted with the same secret key.
In a similar way, if a lingo $\Lambda$ is what below we call
\emph{malleable}, an intruder can disrupt the communication
between an honest sender \emph{Alice} and an honest receiver 
\emph{Bob}
in a protocol $\mathcal{P}$ whose messages are modified
by means of a
lingo\footnote{I.e., honest participants that use a 
\emph{dialect} $\mathcal{D}_{\Lambda}(\mathcal{P})$
based on $\mathcal{P}$ and
$\Lambda$ to modify their messages
(see \S \ref{sec:dialects}).} $\Lambda$ by producing
a message \emph{compliant} with
a secret parameter $a \in A$ of
$\Lambda$ 
supposedly sent from \emph{Alice} to 
\emph{Bob} (who share the secret parameter $a$),
 but actually sent by the intruder.

\begin{definition}[Malleable Lingo] \label{def:mall} Let $\Lambda=(D_1,D_2,A,f,g)$ 
be a lingo,
which has been formally specified as the initial algebra of
an equational theory $(\Sigma_{\Lambda},E_{\Lambda})$.
$\Lambda$ is called
\emph{malleable} iff there exists a \emph{recipe}, i.e., a $\Sigma_{\Lambda}$-term
$t(x,y)$ of sort $D_2$
with free variables $x,y$ of respective sorts $D_2$ and $A$, as well as
 a subset $A_{0} \subseteq A$
 such that
 $\forall d_1 \in D_1, \; \forall a \in A, \; \forall a' \in A_{0}$,
\begin{enumerate}
    \item $f(d_1,a) \not= t(f(d_1,a),a')$, where, by convention,
  $t(f(d_1,a),a')$  abbreviates the substitution instance
  $t\{x \mapsto f(d_1,a),y \mapsto a'\}$

  \item $\exists d'_1 \in D_1$ such that $t(f(d_1,a),a') = f(d'_1,a)$.
\end{enumerate}
%
\end{definition}

\begin{example} \label{ex:xor-mall}
The lingo $\Lambda_{xor}\{n\}$ from
Example \ref{ex:xor-lingo} is \textit{malleable}.  The recipe $t(x,y)$
 is $x \oplus y$,  and $A_{0}= \{0,1\}^{n} \setminus \{\overrightarrow{0}\}$.  Then, for any $a' \in A_{0}$
 we have,
 $t(f(d_1,a),a')=d_{1} \oplus a \oplus a' \not = d_{1} \oplus a $,
 which meets condition $(1)$, since $a' \not= \overrightarrow{0}$ and $\oplus$ is a group addition;
 and it meets condition $(2)$, since
 $t(f(d_1,a),a')=f(d_1 \oplus a',a)$. An attacker needs not know $d_1 \oplus a'$.
\end{example}

\noindent 
The notion of an $f$-checkable lingo of Section~\ref{f-check-subsection}
provides a first line of defense against an intruder trying to
generate a payload compliant with a secret parameter $a$.
However, the example below shows that an $f$-checkable lingo may be malleable.

\begin{example} \label{ex:xor-sharp-mall}
The lingo $\Lambda_{\mathit{xor}}^{\sharp}\{n\}$
from Example \ref{ex:edxor} with $n \geq 2$
is malleable with recipe
$t((x_1 , x_2),(y_1 , y_2))= (x_1 \oplus y_1 , x_2\oplus y_1)$
and $A_{0} = \{(y_1 , y_2) \in \{0,1\}^{n} \otimes \{0,1\}^{n} \mid 
y_1 \not= \overrightarrow{0}\}$.
Indeed, for any $d_1 \in \{0,1\}^{n}$, $(a_1 , a_2) \in  \{0,1\}^{n} \otimes \{0,1\}^{n}$, and $(a'_1 , a'_2) \in A_{0}$, we have:
\begin{enumerate}
\item 
$t(f^{\sharp}(d_1,(a_1,a_2)),(a'_1,a'_2)))=
(d_1 \oplus a_1 \oplus a'_1, d_1 \oplus a_2 \oplus a'_1)$
so that, by $a'_1 \not= \overrightarrow{0}$, we have
$f^{\sharp}(d_1,(a_1,a_2)) \not= t(f^{\sharp}(d_1,(a_1,a_2)),(a'_1,a'_2)))$.
 \item From $t(f^{\sharp}(d_1,(a_1,a_2)),(a'_1,a'_2)))=
(d_1 \oplus a_1 \oplus a'_1, d_1  \oplus a_2 \oplus a'_1) $
it trivially follows that there is a $d'_1$, namely,
 $d'_1 = d_1 \oplus a'_1$, such that
$f^{\sharp}(d'_1,(a_1,a_2))= 
(d_1 \oplus a'_1 \oplus a_1,
d_1 \oplus a'_1 \oplus a_2)=
t(f^{\sharp}(d_1,(a_1,a_2)),(a'_1,a'_2)))$.
\end{enumerate}
Therefore, $\Lambda_{\mathit{xor}}^{\sharp}\{n\}$ with $n \geq 2$ is malleable.
An attacker needs not know $d'_1$, only needs to use
$f^{\sharp}(d_1,(a_1,a_2))$, choose $(a,a') \in A_{0}$, and then use $t$.
\end{example}

\noindent Examples \ref{ex:xor-mall}--\ref{ex:xor-sharp-mall}
suggest the following generalization
of Theorem~\ref{theorem:trasnform-fCheckable} above.
The proof is in Appendix~\ref{app:proofs}.


\begin{theorem}\label{theorem:malleable} 
Let $\Lambda=(D_1,D_2,A,f,g)$ 
be a lingo which has been formally specified as the initial algebra of
an equational theory $(\Sigma_{\Lambda},E_{\Lambda})$ and is
such that: 
(i) $D_{2} \subseteq  D_{1}$, (ii)
there is a 
subset $A_{0} \subseteq A$ with at least two different elements
$a_0 , a_1$ and
such that $\forall d_1 \in D_1$,
$\forall a' \in A_{0}$, $f(d_1,a') \not= d_1$, and
(iii) $\forall \; d_1 \in D_1$, $\forall a \in A$, 
$\forall a' \in A_{0}$, 
$f(f(d_1,a),a')=f(f(d_1,a'),a)$.  Then both $\Lambda$
 and $\Lambda^{\sharp}$   are malleable.
\end{theorem}

\begin{remark} \label{rem:mall1} Note that we don't require a specification of  the relation between the two plaintexts $d_1$ and $d'_{1}$ (as is done in the informal definition given at the beginning of this section)
except for both belonging to $D_1$.   This still captures the way in which an attacker would take advantage of a malleable function but is more straightforward to reason about.
\end{remark}

\begin{remark} \label{rem:mall2} Note that the recipes are used with two arguments: the ciphertext $f(d_1,a)$ and a parameter $a’$ 
in  $A_{0}$.  It is also possible to ignore $f(d_1,a)$, i.e., for some lingos the attacker may generate a random cyphertext that, due to malleablity,
can be mistaken for a genuine message by the receiver. More generally, other notions of lingo malleability besides that in
Def. \ref{def:mall} may be studied. 
\end{remark} 

\begin{remark} \label{rem:mall3}
Malleability is in general undesirable, but some threats introduced by it are mitigated by the on-path  attacker model. For example, the on-path attacker cannot substitute a message $f(d'_1,a)$ for the message $f(d_1,a)$ from which it was derived. But there are also other cases, for example if we allow repeated parameters, or encounter a situation in which an attacker can send a random message which a receiver could mistake for some $f(d_1,a)$, as in Remark \ref{rem:mall2}.
\end{remark}
 

%
%
%

\noindent  
Finally, one would 
like to transform
a  possibly malleable lingo $\Lambda$ into a 
non-malleable lingo $\Lambda'$.  Developing general methods for
proving that a lingo is non-malleable is a topic for future research.  
Note, however, that in cryptography non-malleability of a cryptosystem has been shown to be equivalent to authenticated encryption \cite{bellare1998relations,bellare2008authenticated}. This may or may not be the case for protocol dialects, which use a weaker attacker model.  In Appendix \ref{app:authenticating-lingos}, we discuss \emph{authenticating lingos}
and a transformation mapping a lingo to an authenticating one which we conjecture
can be used to build non-malleable lingos out of malleable ones.

\section{Lingo Compositions} \label{sec:compositions}

\noindent Lingos can be composed in various ways to obtain new lingos. Lingo compositions provide modular, automated methods of obtaining new lingos from existing ones.  Such compositions
are often stronger, i.e., harder to compromise by an attacker, than the lingos so composed. We define here two such lingo composition operations, namely, \emph{horizontal} and \emph{functional} composition of lingos.  

\subsection{Horizontal Composition of Lingos}\label{subsubsec:horizontal-composition}

\noindent Given a finite family of lingos $\overrightarrow{\Lambda} = \{\Lambda_{i}\}_{1 \leq i \leq k}$, $k \geq 2$, all sharing the same input data type $D_1$,
their \emph{horizontal composition}, denoted $\bigoplus \overrightarrow{\Lambda}$
is intuitively
their union, i.e., in $\bigoplus \overrightarrow{\Lambda}$ the input data type
$D_1$ remains the same,
$D_2$ is the union of the $D_{2.i}$, $1 \leq i \leq k$,
and $A$ is the disjoint union of the $A_{i}$, $1 \leq i \leq k$.
$\bigoplus \overrightarrow{\Lambda}$ is a more complex lingo than each of its 
component lingos
$\Lambda_{i}$, making it harder to compromise
by an attacker, because it becomes a hydra with $k$ heads:
quite unpredictably,
sometimes behaves like some $\Lambda_{i}$ 
and sometimes like another $\Lambda_{j}$.  Furthermore,
it has a bigger parameter space, since
$|A| = \Sigma_{1 \leq i \leq k} |A_{i}|$.

\begin{definition}[Horizontal Composition]\label{def:hor-comp}
Let $\overrightarrow{\Lambda}$ be a finite family of lingos
$\overrightarrow{\Lambda} = \{\Lambda_{i}\}_{1 \leq i \leq k}$,
$k \geq 2$, all having the same input data type $D_1$, i.e.,
$\Lambda_{i}=(D_{1},D_{2.i},A_{i},f_{i},g_{i})$, $1 \leq i \leq k$.
Let $\overrightarrow{d}_{0}=(d_{0.1},\ldots,d_{0.k})$ be a choice of default
$D_{2.i}$-values, $d_{0.i}\in D_{2.i}$, $1 \leq i \leq k$.
Then, the \emph{horizontal composition} of the lingos
$\overrightarrow{\Lambda}$ with default $D_{2.i}$-values  $\overrightarrow{d}_{0}$
is the lingo:
\[
\bigoplus_{\overrightarrow{d_{0}}} 
\overrightarrow{\Lambda}=(D_{1}, \bigcup_{1 \leq i \leq k} D_{2.i},
\bigcup_{1 \leq i \leq k} A_{i} \times\{i\}
,\oplus \overrightarrow{f},\oplus \overrightarrow{g})
\]
where, for each $d_1 \in D_1$, $d_2 \in \bigcup_{1 \leq i \leq k} D_{2.i}$, and $a_{i} \in A_{i}$, $1 \leq i \leq k$,
\begin{enumerate}
\item $\oplus \overrightarrow{f} = \lambda (d_1,(a_i ,i)).f_i(d_1,a_i)$,
\item $\oplus \overrightarrow{g} = \lambda (d_2 ,(a_i ,i)).\; 
\mathbf{if} \; d_{2} \in D_{2.i} \; \mathbf{then} \; g_{i}(d_i ,a_i) \; \mathbf{else} \; g_{i}(d_{0.i},a_{i})$.
\end{enumerate}
$\bigoplus_{\overrightarrow{d_{0}}} \overrightarrow{\Lambda}$ is indeed a lingo:
$\oplus \overrightarrow{g}(\oplus \overrightarrow{f}(d_1,(a_{i},i)),(a_i , i))=
g_{i}(f_{i}(d_1,a_i ),a_i )= d_1$.
%
\end{definition}

\vspace{2ex}

\noindent We have already pointed out that, in practice, two
protocol participants using a lingo to first modify and then decode a payload $d_1$ that, say, Alice sends to Bob, agree on a secret parameter
$a \in A$ by agreeing on a secret random number $n$,
because both use a function $\mathit{param}: \mathbb{N} \rightarrow A$
to get $a = \mathit{param}(n)$.  This then poses the practical
problem of synthesizing a function $\oplus \overrightarrow{\mathit{param}}: \mathbb{N} \rightarrow \bigcup_{1 \leq i \leq k} A_{i} \times \{i\}$
for $\bigoplus_{\overrightarrow{d_{0}}} \overrightarrow{\Lambda}$
out of the family a functions $\{ \mathit{param}_{i} : \mathbb{N} \rightarrow A_{i}\}_{1 \leq i \leq k}$ used in each of the lingos $\Lambda_{i}$.
Furthermore, different lingos in the family $\overrightarrow{\Lambda}$
may have different degrees of strength against an adversary.
This suggests favoring the choice of stronger lingos over that
of weaker lingos in the family $\overrightarrow{\Lambda}$ to achieve a
function $\oplus \overrightarrow{\mathit{param}}$ that improves the overall strength of 
$\bigoplus_{\overrightarrow{d_{0}}} 
\overrightarrow{\Lambda}$.  This can be achieved by choosing a \emph{bias vector}
$\overrightarrow{\beta}=(\beta_1,\ldots,\beta_k) \in \mathbb{N}_{>0}^{k}$,
so that, say, if lingo $\Lambda_i$ is deemed to be stronger than lingo $\Lambda_j$, then 
the user chooses $\beta_i > \beta_j$.  That is, $\overrightarrow{\beta}$
specifies a \emph{biased} die with $k \geq 2$ faces,
so that the die will show face $j$ with probability $\frac{\beta_{j}}{\Sigma_{1 \leq i \leq k} \beta_{i}}$.  Therefore, we can use a
pseudo-random function $\mathit{throw}_{\overrightarrow{\beta}}: \mathbb{N} \rightarrow \{1,\ldots, k\}$ simulating a sequence of throws of a $k$-face die with bias $\overrightarrow{\beta}$
to get our desired function 
$\oplus \overrightarrow{\mathit{param}}$ as the function:
\[
\oplus \overrightarrow{\mathit{param}}(n) = (\mathit{param}_{\mathit{throw}_{\overrightarrow{\beta}}(n)}(n),\mathit{throw}_{\overrightarrow{\beta}}(n)).
\]

\begin{example} \label{ex:dccomp} To see how horizontal composition can strengthen lingos, consider the D\&C lingo described in Example \ref{ex:d&c-lingo-old}, in which $ f(n,a) = (x,y)$  where $x = quot(n +(a+2),a +2)$ and $y = rem(n + (a+2),a+2)$. We note that any choice of $0 \le y < a + 2$ will pass the $f$-check, so the choice of $y = 0$ or $1$ will always pass the $f$-check. To counter this,   let \emph{reverse}-D\&C be the  lingo with $f(n,a) = (y,x)$, where $y$ and $x$ are computed as in D\&C.  The attacker's best strategy in reverse-D\&C is the opposite of that in D\&C.  Thus composing D\&C  and reverse-D\&C  horizontally with a bias vector $(.5,.5)$ means that the attacker strategy of choosing the first (respectively, second) element of the output to be 0 succeeds with probability 0.5 in any particular instance, as opposed to probability 1 for D\&C by itself.
\end{example}

\noindent 
The following result follows easily from the definition of horizontal composition.

\begin{lemma}
If each $\Lambda_i$ in $\overrightarrow{\Lambda} = \{\Lambda_{i}\}_{1 \leq i \leq k}$,
$k \geq 2$, is $f$-checkable, then $\bigoplus_{\overrightarrow{d_{0}}} 
\overrightarrow{\Lambda}$ 
is also $f$-checkable.
\end{lemma}

\begin{example}[Horizontal composition of XOR-BSeq and D\&C Lingos]
\label{ex:horizontal-comp}
The lingos $\Lambda_{xor.\mathit{BSeq}}$
of Example \ref{ex:XOR-BSeq-Lingo}
and $\lingo_{D\&C}$ of Example \ref{ex:d&c-lingo-old}
share the same $D_{1}$, namely, $\mathbb{N}$.
Therefore, they do have a horizontal composition
$\Lambda_{xor.\mathit{BSeq}} \oplus_{\overrightarrow{d}_{0}} \lingo_{D\&C}$
for any choice of $\overrightarrow{d}_{0}$.  

We refer the reader to Appendix~\ref{app:ex-lingo-horizontal-comp} for details on the executable specification of horizontal composition in Maude.
\end{example}

\subsection{Functional Composition of Lingos}\label{subsubsec:functional-composition}

Given two lingos $\Lambda$ and $\Lambda'$ such that the output data type
$D_2$  of $\Lambda$ coincides with the input data type of $\Lambda'$,
it is possible to define a new lingo $\Lambda \odot \Lambda'$
whose $f$ and $g$ functions are naturally the compositions
of $f$ and $f'$ (resp. $g'$ and $g$) in a suitable way.
Moreover, $\Lambda$  and $ \Lambda'$ can be chosen so that $\Lambda \odot \Lambda'$ is  harder to
compromise than either $\Lambda$ or $\Lambda'$. 

\begin{definition}[Functional Composition]\label{def:fun-comp}
Given lingos $\Lambda=(D_1,D_2,A,f,g)$ and $\Lambda'=(D_2,D_3,A',f',g')$, 
their \emph{functional composition} is the lingo
$\Lambda \odot \Lambda' =
(D_1,D_3,A \times A',f \cdot f',g * g')$, where
for each $d_1 \in D_1$, $d_3 \in D_3$, and 
$(a,a') \in A \times A'$,
\begin{itemize}
\item $f \cdot f'(d_1,(a,a')) =_{\mathit{def}}
f'(f(d_1,a),a')$,
\item $g * g'(d_3,(a,a')) =_{\mathit{def}}
g(g'(d_3,a'),a)$.
\end{itemize}
$(D_1,D_3,A \times A',f \cdot f',g * g')$ is indeed a lingo,
since we have:
\begin{align*}
    g * g'(f \cdot f'(d_1,(a,a')),(a,a')) &= g(g'(f'(f(d_1,a),a'),a'),a) \\
    &=g(f(d_1,a),a) = d_1.
\end{align*}
\end{definition}

\begin{lemma}\label{f-check-4-fun-lemma}
Given lingos $\Lambda=(D_1,D_2,A,f,g)$ and $\Lambda'=(D_2,D_3,A',f',g')$, if $\Lambda'$ is $f$-checkable,
then 
$\Lambda \odot \Lambda'$ is also $f$-checkable.
\end{lemma}

\begin{example}[Functional composition of XOR-BSeq and D\&C Lingos]
\label{ex:functional-comp}
The lingos $\Lambda_{xor.\mathit{BSeq}}$
of Example \ref{ex:XOR-BSeq-Lingo}
and $\lingo_{D\&C}$ of Example \ref{ex:d&c-lingo-old}
are functionally composable as
$\Lambda_{xor.\mathit{BSeq}} \odot \lingo_{D\&C}$,
because $\mathbb{N}$ is both the output
data type of $\Lambda_{xor.\mathit{BSeq}}$ and
the input data type of $\lingo_{D\&C}$.
%
%
Note that, by Lemma \ref{f-check-4-fun-lemma}
above, and Theorem \ref{theorem:trasnform-fCheckable} above,
$\Lambda_{xor.\mathit{BSeq}} \odot \lingo_{D\&C}$
is $f$-checkable, in spite of $\Lambda_{xor.\mathit{BSeq}}$
not being so. 
\end{example}

We refer the reader to Appendix~\ref{app:ex-lingo-func-comp} for details on the executable specification of functional composition in Maude.

\section{Lingo Transformations and Dialects as Formal Patterns}\label{sec:dialects}

We make the notion of a \emph{formal pattern} as a transformation of rewrite theories explicit and explain how
it applies to both lingos and dialects.  Since dialects are parametric on the chosen lingo and the chosen protocol, we
begin by explaining how protocols are formalized as generalized actor rewrite theories.  With all these notions in hand we
can explain in detail the \emph{lingo formal pattern} and its rewriting logic semantics and give examples of different dialects for a given
protocol based on the choice of various lingos, sometimes themselves compositions of simpler ones.  Finally, we explain
how both the lingo formal pattern, lingos and dialect formal patterns are specified as \emph{formal executable specifications}
in the Maude rewriting logic language \cite{maude-book}, a theme further developed in Appendix \ref{app:examples}.

\subsection{Protocols as Generalized Actor Rewrite Theories}

\emph{Actors} \cite{Agha86} are a widely used model for distributed systems, in which  concurrent objects communicate through asynchronous message passing. When an actor consumes a received message it can change its state, send new messages, and create new actors. 
Actor systems can be formally specified and executed as \emph{actor rewrite theories} \cite{ooconc}.
In \emph{rewriting logic} \cite{MeseguerTCS92} a concurrent system can be formally specified and executed as a \emph{rewrite theory} 
$\mathcal{R}=(\Sigma,E,R)$, with $(\Sigma,E)$ an equational theory with function symbols $\Sigma$ and equations $E$
that specifies the concurrent system's \emph{states} as elements of the initial algebra (algebraic data type) $\mathbb{T}_{\Sigma/E}$
of $(\Sigma,E)$, and where the system's \emph{local concurrent transitions} are specified by \emph{rewrite rules} in $R$
of the form $u(x_{1},\ldots,x_{n}) \rightarrow v(x_{1},\ldots,x_{n})$, where $u(x_{1},\ldots,x_{n})$ is a $\Sigma$-expression identifying those local state fragments
that can make a local transition to a state fragment of the form $v(x_{1},\ldots,x_{n})$.  Under natural assumptions, a rewrite theory $\mathcal{R}=(\Sigma,E,R)$
is \emph{executable} by rewriting.  Maude \cite{maude-book} is a declarative programming language whose programs, called \emph{system modules}, are rewrite theories
declared as \texttt{mod} $(\Sigma,E,R)$ \texttt{endm}.  Maude has a functional sublanguage of \emph{functional modules},
which are equational theories declared as \texttt{fmod} $(\Sigma,E)$ \texttt{endfm}, specifying an algebraic data type  $\mathbb{T}_{\Sigma/E}$.
For example, (more on this in Section \ref{maude-subsection}) lingos are algebraic data types that can be naturally specified as functional modules in Maude.

Any \emph{communication protocol} $\mathcal{P}$ can be formally specified and programmed as a \emph{generalized actor rewrite theory}
 $\mathcal{P}=(\Sigma,E \cup B,R)$ \cite{DBLP:journals/pacmpl/LiuMOZB22}.  The rewrite rules of an actor rewrite theory are of the form:
\[(\mathit{to}\; B\; from \; A : M) \; <\, B : Cl\, |\, \mathit{atts}\, >  \;
\rightarrow \; <\, B : Cl\, |\, \mathit{atts}'\, > \;\; \mathit{msgs}  \;\; \mathit{actors}.
\]
They specify how actor $B$ of class $Cl$, upon receiving a message from actor $A$, can change its local state (the attribute-value pairs $\mathit{atts}$)
to $\mathit{atts}'$ and may generate several new messages $\mathit{msgs}$ and several new actors $\mathit{actors}$.
Any such rule specifies an \emph{asynchronous} local concurrent transition in a distributed state (called a \emph{configuration})
which is a multiset of objects and messages built up with a multiset union operation $\_\;\_$ (juxtaposition) satisfying the axioms
$B$ of associativity, commutativity and empty multiset $\emptyset$ as unit element \cite{ooconc}.  Generalized actor rewrite theories allow
also rules of the form:
\[<\, B : Cl\, |\, \mathit{atts}\, >  \;
\rightarrow \; <\, B : Cl\, |\, \mathit{atts}'\, > \;\; \mathit{msgs}  \;\; \mathit{objs}.
\]
describing \emph{active actors} that can also change their state and generate new messages and actors without receiving a message.

Because of space limitations, we refer the reader to Appendix~\ref{app:FormPatts} for details on how lingos can be used as Formal Patterns.

\vspace{-1ex}
\subsection{The Dialect Transformation as a Formal Pattern} \label{Dialects-sub}

The \emph{dialect transformation} can be formalized as a formal pattern of the form:
\[
\mathcal{D}: \mathit{GralActorRewThs}\times \mathit{PLingos}\ni (\mathcal{P},\Lambda_{p}) \mapsto \mathcal{D}_{\Lambda_{p}}(\mathcal{P}) \in \mathit{GralActorRewThs}
\]
where $\mathcal{P}$ is a protocol, i.e.,  a  rewrite theory in the class\footnote{Hereafter we assume generalized actor systems
whose actors do not create new actors.}
$\mathit{GralActorRewriteThs}$ of
generalized actor rewrite theories, and $\mathcal{D}_{\Lambda_{p}}(\mathcal{P})$ is another protocol 
obtained from $\mathcal{P}$ by means of $\Lambda_{p}$, where $\Lambda_{p}$ is
a \emph{lingo with parameter function} $p$ in the class $\mathit{PLingos}$, that is, $\Lambda_{p}$ adds to a lingo
$\Lambda$ a parameter function $p: \mathbb{N} \rightarrow A$.

\begin{figure}[t]
    \centering
    \includegraphics[width=\linewidth]{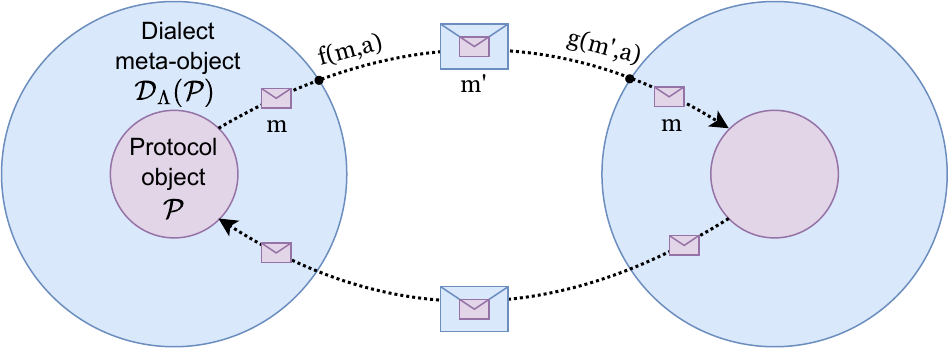}
\vspace{-4ex}    
    \caption{Dialect meta-objects (enhanced from \cite{ProtocolDialects-FormalPatterns}).}
    \label{fig:dialect-meta-objects}
\vspace{-4ex}    
\end{figure}

The essential ideas about the dialect $\mathcal{D}_{\Lambda_{p}}(\mathcal{P})$ are summarized in Figure \ref{fig:dialect-meta-objects}.
Suppose that Alice and Bob are actors in the protocol $\mathcal{P}$ and Alice sends payload $m$ to Bob in a message.
What happens instead in $\mathcal{D}_{\Lambda_{p}}(\mathcal{P})$ is that \emph{Alice and Bob behave as before},
that is, their original protocol $\mathcal{P}$ remains the same.  However, unbeknownst to Alice and Bob,
their communication is now \emph{mediated} by their corresponding \emph{meta-actors} (also named Alice, resp. Bob)
inside which they are now \emph{wrapped}.  Alice's meta-actor
\emph{transforms} the original payload $m$ into $f(m,a)$, and Bob's meta-actor \emph{decodes} $f(m,a)$ into $g(f(m,a),a)=m$
using $\Lambda_{p}$ and
the \emph{shared secret parameter} $a$. Of course, if Bob replies to Alice with another message, the
roles are exchanged: now Bob's meta-actor will transform the new payload $m'$ to $f(m',a')$
and Alice's meta-actor with get back $m'$ because they now share a \emph{new} secret parameter $a'$.
Since parameters change all the time, dialects become \emph{moving targets} for an attacker.

\begin{figure*}[t]
\begin{align*}
\mathit{in}\; &:\; < O_1 : \textit{DC} \mid \mathit{in.buffer} : M_{in}, \mathit{atts}'>\ (to \ O_1 \ from \ O_2 \ : \ P) \\
\rightarrow\ &< O_1 : \textit{DC} \mid \mathit{in.buffer} : (M_{in} \cup (to \ O_1 \ from \ O_2 \ : \ P), \mathit{atts}' > \ 
\\[1ex]
\mathit{out}\; &:\;
< O_1 : \mathit{DC} \mid \mathit{conf} : (< O_1 : C \mid \mathit{atts} >\ (to \ O_2 \ from \ O_1 \ \texttt{:} \ P) \cup \ M),\ \\&\quad\quad\quad\quad\quad\quad\;\mathit{peer.counters}:R ,\ \mathit{atts}'> \\
\rightarrow&\ < O_1 : \mathit{DC} \mid \mathit{conf} : (< O_1 : C \mid \mathit{atts} >\ M \setminus (to \ O_2 \ from \ O_1 \ : \ P)),\\
&\quad\quad\quad\quad\quad\quad\operatorname{update}(\mathit{peer.counters}:R, \mathit{atts}') > \\
&\quad (to \ O_2 \ from \ O_1\ : \ f(P, param(R[O_2]))) 
\\[1ex]
\mathit{deliver}\; &:\;
< O_1 : \textit{DC} \mid \mathit{conf} : (< O_1 : C \mid \mathit{atts} > M),\ \mathit{in.buffer} : M_{in},\\
&\quad\quad\quad\quad\quad\quad\;\mathit{peer.counters}:R ,\ \mathit{atts}'> \\
\rightarrow&\ < O_1 : \textit{DC} \mid \mathit{conf} : (< O_1 : C \mid \mathit{atts} > M \cup \{g(M_{selected}, \operatorname{param}(R[O_2]))\},\\ 
&\quad\quad\quad\quad\quad\quad\mathit{in.buffer} : (M_{in} \setminus M_{selected}),\ \\
&\quad\quad\quad\quad\quad\quad {update}(\mathit{peer.counters}:R, \mathit{atts}') > \\
\mathit{if}&\ size(M_{in}) \geq egressArity \; \land M_{selected} := take(egressArity, M_{in}) \; \land\\
&fromOidTag(M_{selected}, R) = true \; \texttt{.}
\end{align*}
\vspace{-4ex}
\caption{Meta-Actor Rewrite Rules.}
\label{fig:wrapper-actor-rules}
\vspace{-4ex}
\end{figure*}

But what is a \emph{meta-actor}?    In rewriting logic a meta-actor of class $\mathit{Dialect}$ will have the general form:
$<\, o : \mathit{Dialect}\, |\, \mathit{conf}:\, <\, o : C\, |\, \mathit{atts}\, >\, ,\mathit{datts} >$.
That is, one of the attributes in the state of the meta-actor is called $\mathit{conf}$
and consists of a \emph{configuration} of objects and messages with a \emph{single object},
namely, $<\, o : C\, |\, \mathit{atts}\, >$ and, at various times, either an incoming message addressed to
$o$, or outgoing messages sent by $o$ to other actors. 

$\mathit{datts}$ denotes the rest of the
dialect meta-actor's state (i.e., it remaining attribute-value pairs).  What are they?
Besides the object attribute $\mathit{conf}$, $Dialect$ meta-actors have the following additional attributes:
(i) $\mathit{in.buffer}$: a buffer of incoming messages to which the lingo's decoding function $g$ has not yet been applied;
and
(ii) $\mathit{peer.counters}$: a map from actor names to natural numbers that serves as a counter for the number of messages exchanged. It is also used 
to generate (using the parameter function $p$) the new secret parameters used by the lingo's $f$ and $g$ functions for communicating with other actors.

The \emph{behavior} of a \emph{Dialect} meta-actor is specified by the rewrite rules in
 Figure~\ref{fig:wrapper-actor-rules}.  These can be summarized as follows:
 (i) rule $\operatorname{\textit{in}}$ collects incoming messages for the wrapped actor in the $\mathit{in.buffer}$ attribute; (ii) rule $\operatorname{\textit{out}}$ processes messages produced by the wrapped protocol, sending the transformed messages into the network, by applying the lingo's $f$ function with the corresponding parameter; 
 and
 (iii) rule $\operatorname{\textit{deliver}}$ processes messages stored in the buffer, by applying the lingo's $g$ function with the proper parameter, delivering the original version of the message to the protocol.
 Rules $\operatorname{\textit{out}}$ and $\operatorname{\textit{deliver}}$ apply $f$ and $g$, respectively, and supply the parameter $a$ by calling the lingo's $\operatorname{param}$ function, which computes a parameter $a \in A$ from a number $n \in \mathbb{N}$. The dialect meta-object keeps parameters changing by supplying to the $\operatorname{param}$ function the current message number from the map in $\mathit{peer.counter}$. This turns dialects into moving targets.


 \begin{example}[Transforming protocols by means of Protocol Dialects]\label{ex:coealescing}
     Let us leverage the various lingos introduced in previous sections to incrementally construct more sophisticated, and possibly more secure, protocols through the use of dialects. 
     We use as core protocol a parametric version of MQTT, denoted $Mqtt\{D\}$ and explained in Appendix~\ref{app:ex-protocols}, where $D$ represents the data-type of the message payloads (e.g., bit-sequences of $n \in \mathbb{N}$). Furthermore, we denote any given lingo as $\Lambda_{id}\{D\}$, where $id$ identifies the lingo and $D$ indicates the input data type. For instance, $\Lambda_{D\&C}\{BitSEQ\{N\}\}$ refers to the divide and check lingo over bit-sequences of length $n \in \mathbb{N}$.
     
     As a first transformation, we can use the lingo presented in Example~\ref{ex:xor-lingo}. To do so, we can instantiate a dialect $\mathcal{D}_{\Lambda_p}(\mathcal{P})$ over MQTT, using the xor lingo over a payload consisting of 1 byte, i.e. a bit-sequence (or vector) of length 8. The resulting protocol is denoted by $\mathcal{D}_{\Lambda_{XOR}\{BitSEQ\{8\}\}}(Mqtt\{BitSEQ\{8\}\})$. 

     It is important to note that we can, in a very intuitive way, \textit{upgrade} or adapt the new protocol to meet specific needs. For instance, we can increase the parameter size by just increasing the input parameter. This leads to a new protocol such as $\mathcal{D}_{\Lambda_{XOR}\{BitSEQ\{256\}\}}(Mqtt\{BitSEQ\{256\}\})$, which operates over 32-byte payloads.

     A similar approach can be applied to both parameters of a Dialect: the lingo and the underlying protocol. For instance, consider replacing the XOR\{N\} lingo for the XOR-BSEQ lingo introduced in Example~\ref{ex:XOR-BSeq-Lingo}. These kind of substitutions are valid provided that both the input data type of the lingo and protocol coincide. For this purpose, we can instantiate the protocol's data type to be the natural numbers, that is bit sequences of arbitrary length. Accordingly, the resulting protocol can be expressed as $\mathcal{D}_{\Lambda_{xor.BSeq}}(Mqtt\{Nat\})$. Furthermore, we can change the lingo to the divide and check lingo given in Example~\ref{ex:d&c-lingo} in a straightforward way since both XOR-BSEQ and D\&C share input data types. The resulting protocol is $\mathcal{D}_{\Lambda_{D\&C}}(Mqtt\{Nat\})$. 
     
     
     Furthermore, we can use lingo transformations to construct more sophisticated dialects. For example, one may horizontally compose XOR-BSEQ with D\&C, as illustrated in Example~\ref{ex:horizontal-comp}. Then, we can use this new lingo to yield a new protocol defined as $\mathcal{D}_{\Lambda_{xor.BSeq} \oplus_{\overrightarrow{d}_{0}} \Lambda_{D\&C}}(Mqtt\{Nat\})$. Another protocol can be obtained by using the functional composition from Example~\ref{ex:functional-comp}. The resulting protocol would be $\mathcal{D}_{\Lambda_{xor.BSeq} \odot \Lambda_{D\&C}}(Mqtt\{Nat\})$. For a more detailed description of all these instantiations, we refer the reader to Appendix~\ref{app:ex-dialects}.
 \end{example}

 Last but not least, this work provides a substantial extension and simplification of the previous theory of dialects in \cite{ProtocolDialects-FormalPatterns} that includes,
 not only new lingo properties and formal patterns, but the \emph{replacement} of the former \emph{horizontal}, resp. \emph{vertical}, dialect composition
 operations  in \cite{ProtocolDialects-FormalPatterns} by the simpler and more efficient notions of \emph{horizontal}, resp. \emph{functional}, lingo composition operations.

\vspace{-1ex}
\subsection{Lingos, Dialects, and their Formal Patterns in Maude} \label{maude-subsection}

As shown in Appendix~\ref{app:FormPatts} and Section~\ref{Dialects-sub}, dialects can be formalized as equational theories with initial semantics, and
dialect transformations as formal patterns composing or transforming them.  Likewise, we have seen that a dialect $\mathcal{D}_{\Lambda_{p}}(\mathcal{P})$
is a formal pattern transforming a generalized actor rewrite theory $\mathcal{P}$ into another generalized actor rewrite theory.  Since Maude programs are
exactly either equational theories with initial semantics, or rewrite theories, both lingos and
dialects, as well as their transformations, can be made executable in Maude. 
In fact, all the theoretical
concepts and examples presented in this paper have been formally specified in Maude at the same time that their theoretical properties
were proved.  But the interest of these formal specifications is not just theoretical.  On the one hand, the
$D$-transformation proposed in \cite{D-transf-NASA} provides a formal pattern to automatically transform generalized
actor rewrite theories into distributed implementations; on the other hand, dialects formalized in Maude can be
formally analyzed by various model checking verification techniques, both in Maude itself \cite{maude-book} and, for probabilistic
properties, in the QMaude tool \cite{rubio2023qmaude}.

\vspace{-1ex}
\section{Concluding Remarks}\label{sec:conclusion}
\vspace{-1ex}

Dialects are a resource-efficient approach as a first line of defense against outside attackers. Lingos are invertible message transformations used by dialects to 
thwart attackers from maliciously disrupting communication. In this paper we propose new kinds of lingos and make them stronger through lingo transformations, 
including composition operations, all of them new. 
These lingo transformations and composition
operations are \emph{formal patterns} with desirable security properties.
We also propose a refined and simpler definition for dialects by replacing two former dialect composition operations in \cite{ProtocolDialects-FormalPatterns} by considerably 
simpler and more efficient lingo composition operations.  
All these concepts and transformations have been formally specified and made executable in Maude.

Our next step is to explore lingos in action as they would be used by dialects in the face of an on-path attacker.  We have been developing a formal intruder model that, when combined with the dialect specifications provided in this paper, can be used for statistical and probabilistic model checking of dialects. This model provides the capacity to specify the probability that the intruder correctly guesses secret parameters.
For this we are using QMaude\cite{rubio2023qmaude}, which supports probabilistic and statistical model checking analysis of Maude specifications.



\subsubsection*{Acknowledgements}
S.~Escobar and V.~Garc\'{\i}a have been partially supported 
by the grant CIPROM/2022/6 funded by Generalitat Valenciana
and
by the grant PID2024-162030OB-100 funded by MCIN/AEI/10.13039/501100011033 and ERDF A way of making Europe. We also thank Zachary J Flores for his comments.

%
%
%
\bibliographystyle{splncs04}
\bibliography{biblio/dialects, biblio/maude}

\appendix

\section{Proofs of Lingo Properties}\label{app:proofs}

\noindent This Appendix provides proofs for useful properties of lingos, and lingo transformations and composition operations, stated in the paper.

\begin{corollary}\label{cor:corollary1}
Let $\Lambda=(D_1,D_2,A,f,g)$ be a lingo. Then, $\forall d_1 \in D_1$, $\forall d_2 \in D_2$, $\forall a \in A$, $d_2 = f(d_1, a) \Rightarrow d_2 = f(g(d_2, a),a)$. 
\end{corollary}

\noindent {\bf Proof}:
By Lemma \ref{lem:lemma2} in Appendix~\ref{Aux-Lemmas}, $d_1 = g(d_2,a)$. Thus, $d_2 = f(d_1,a) = f(g(d_2,a),a)$.%
$\Box$

\begin{corollary}\label{cor:corollary2}
Let $\Lambda=(D_1,D_2,A,f,g)$ be a lingo. Then, $\forall d_2 \in D_2$,  $\forall a \in A$, $((\exists d_1 
\in D_1, \; f(d_1, a) = d_2) \Rightarrow d_2 = f(g(d_2, a),a))$. 
\end{corollary}

\noindent {\bf Proof}: By Lemma \ref{lem:lemma2} in Appendix~\ref{Aux-Lemmas} we have,
$\forall d_2 \in D_2,\forall a \in A, (\forall d_1 \in D_1, d_2 = f(d_1,a) \Rightarrow d_2 = f(g(d_2,a),a)),$
which means
$\forall d_2 \in D_2,\forall a \in A, (\forall d_1 \in D_1, \;
\neg(d_2 = f(d_1,a)) \vee d_2 = f(g(d_2,a),a))$
which by $d_1$ not a free variable in $d_2 = f(g(d_2,a),a)$
is equivalent to
$\forall d_2 \in D_2,\forall a \in A, ((\forall d_1 \in D_1, \;
\neg(d_2 = f(d_1,a))) \vee d_2 = f(g(d_2,a),a))$
which by duality of quantifiers is equivalent to
$\forall d_2 \in D_2,\forall a \in A, (\neg (\exists d_1 \in D_1, \;
d_2 = f(d_1,a)) \vee d_2 = f(g(d_2,a),a))$
which means
$\forall d_2 \in D_2,\forall a \in A, ((\exists d_1 \in D_1, \;
d_2 = f(d_1,a)) \rightarrow d_2 = f(g(d_2,a),a))$
as desired.%
$\Box$

\vspace{1ex}

Corollary \ref{cor:corollary3} below shows that
users of a lingo can check if a received payload $d_2 \in D_2$ 
is of the form $d_2 = f(d_1 , a)$ for 
some $d_1 \in D_1$ and
$a \in A$
the current parameter, i.e., if it was generated
from some actual $d_1 \in D_1$ as $d_2 = f(d_1,a)$.
We call a $d_2 \in D_2$ such that $d_2 = f(d_1,a)$ for some
$d_1 \in D_1$ \emph{compliant} with parameter $a \in A$.
As shown in Corollary \ref{cor:corollary3}
compliance with $a$ can be checked by checking the equality $f(g(d_2,a),a) = d_2$ for any given $d_2 \in D_2$ and $a \in A$.

\begin{corollary}\label{cor:corollary3}
Let $\Lambda=(D_1,D_2,A,f,g)$ be a lingo. Then, $\forall d_2 \in D_2$,  $\forall a \in A$, $\exists d_1, (f(d_1, a) = d_2) \Leftrightarrow d_2 = f(g(d_2, a),a)$. 
\end{corollary}

\noindent {\bf Proof}:
The $(\Rightarrow)$ implication follows
from Corollary \ref{cor:corollary2}.
The $(\Leftarrow)$ implication follows
by choosing 
$d_{1} =g(d_2, a)$.  $\Box$

\vspace{1ex}
\noindent \textbf{Theorem \ref{theorem:trasnform-fCheckable}.}
Let $\Lambda=(D_1,D_2,A,f,g)$ be a lingo, where
    we assume that $D_1$, $D_2$ and $A$ are computable data-types and $|D_{1}| \geq 2$.
    %
    %
    Then, $\Lambda^{\sharp}=(D_1,\allowbreak {D_2}\times{D_2},\allowbreak A \otimes A,\allowbreak f^{\sharp},\allowbreak g^{\sharp})$ is an
    $f$-checkable lingo, where:
    \begin{itemize}
        \item $A \otimes A = A \times A \setminus id_A$, with $id_A = \{(a,a) \in A^2 | a \in A\}$
        \item $f^{\sharp}(d_1, (a,a')) = (f(d_1,a), f(d_1,a'))$
        \item $g^{\sharp}((d_2,d'_2), (a,a')) = g(d_2,a)$
    \end{itemize}
    
\vspace{2ex}

\noindent {\bf Proof}:
We first prove that 
$\Lambda^{\sharp}=(D_1,{D_2}^2,A \otimes A,f^{\sharp},g^{\sharp},\mathit{comp}^{\sharp})$ satisfies the lingo equation
$g(f(d_1,a),a)= d_1$.  Indeed,
$g^{\sharp}(f^{\sharp}(d_1,(a,a')),(a,a')) = g^{\sharp}((f(d_1,a), \\f(d_1,a')),(a,a')) = g((f(d_1,a), a) = d_1$
as desired.

\vspace{1ex}

\noindent To prove that, in addition,
$\Lambda^{\sharp}=(D_1,{D_2}^2,A \otimes A,f^{\sharp},g^{\sharp})$
is an $f$-checkable lingo, note that, since $|D_{1}| \geq 2$,
for each $(a,a') \in A \otimes A$ we can choose $d^{\sharp}_{2} = 
(f(d_1 ,a),f(d'_1,a'))$ with $d_1 \not= d'_1$.
Then, there is no $d''_1 \in D_1$
such that $d^{\sharp}_{2} = f^{\sharp}(d''_1,(a,a'))$,
since this would force $(f(d''_1 ,a),f(d''_1,a'))=(f(d_1 ,a),f(d'_1,a'))$,
which by Lemma \ref{lem:lemma1} in Appendix~\ref{Aux-Lemmas} would force $d''_1 = d_1 = d'_1$, which
is impossible.
$\Box$

\vspace{1ex}
\noindent \textbf{Theorem \ref{theorem:malleable}.}
Let $\Lambda=(D_1,D_2,A,f,g)$ 
be a lingo which has been formally specified as the initial algebra of
an equational theory $(\Sigma_{\Lambda},E_{\Lambda})$ and is
such that: 
(i) $D_{2} \subseteq  D_{1}$, (ii)
there is a 
subset $A_{0} \subseteq A$ with at least two different elements
$a_0 , a_1$ and
such that $\forall d_1 \in D_1$,
$\forall a' \in A_{0}$, $f(d_1,a') \not= d_1$, and
(iii) $\forall \; d_1 \in D_1$, $\forall a \in A$, 
$\forall a' \in A_{0}$, 
$f(f(d_1,a),a')=f(f(d_1,a'),a)$.  Then both $\Lambda$
 and $\Lambda^{\sharp}$   are malleable. \\


\noindent {\bf Proof}:  To see that $\Lambda$ is malleable,
choose as $A_{0}$ the subset in (ii), and as recipe
$t(x,y)$ the term $f(x,y)$, which makes sense by (i).
Then, $\forall d_1 \in D_1$, $\forall a \in A$, $\forall a' \in A_{0}$
we have: (1) $t(f(d_1,a),a') = f(f(d_1,a),a') \not= f(d_1,a)$
by (ii); and (2) there exists $d'_{1}$, namely,
$d'_1 = f(d_1 , a')$ such that, by (iii),
$t(f(d_1,a),a') = f(f(d_1,a),a') = f(f(d_1,a'),a)=
f(d'_1,a)$.

\vspace{1ex}

\noindent To see that $\Lambda^{\sharp}$ is malleable,
define $A^{\sharp} = \{(a_1,a_2) \in A \otimes A \mid a_1 \in A_{0} \}$
and use the recipe $t((x_1 , x_2),(y_1 , y_2))= (f(x_1,y_1), f(x_2,y_1))$.
Then, $\forall d_1 \in D_1$,
$\forall (a_1,a_2) \in A \otimes A$, 
$\forall (a'_1,a'_2) \in A^{\sharp}_{0}$
we have:
(1) $t(f^{\sharp}(d_1,(a_1,a_2)),(a'_1,a'_2))=
(f(f(d_1,a_1),a'_1),\\f(f(d_1,a_2),a'_1)) \not= (f(d_1,a_1),f(d_1,a_2))=
f^{\sharp}(d_1,(a_1,a_2))$ by (ii); and (2)
we can choose $d'_1 = f(d_1,a'_1)$ such that, by (iii),
$t(f^{\sharp}(d_1,(a_1,a_2)),(a'_1,a'_2))=
(f(f(d_1,a_1),\\a'_1),f(f(d_1,a_2),a'_1))= (by 
(f(f(d_1,a'_1),a_1),f(f(d_1,a'_1),a_2))=f^{\sharp}(d'_1,(a_1,a_2)$,
as desired. $\Box$\\

\noindent \textbf{Lemma \ref{f-check-4-fun-lemma}.}
Given lingos $\Lambda=(D_1,D_2,A,f,g)$ and $\Lambda'=(D_2,D_3,A',f',g')$, if $\Lambda'$ if $f$-checkable,
then the \emph{functional composition} 
$\Lambda \odot \Lambda'$ is also $f$-checkable.

\vspace{2ex}

\noindent {\bf Proof}: We need to show that for each $(a,a') \in A \times A'$
there exists $d_3 \in D_3$ such that $\nexists d_1 \in D_1$ such that
$d_3 = f'(f(d_1,a),a')$.  But, by assumption, there exists
a $d_3 \in D_3$ such that $\nexists d_2 \in D_2$ such that
$d_3 = f'(d_2,a')$.  This fact applies, in particular,
to $d_2 = f(d_1,a)$.  $\Box$

\subsection{Auxiliary Lemmas Used in the Above Proofs} \label{Aux-Lemmas}

\begin{lemma}\label{lem:lemma1}
Let $\Lambda=(D_1,D_2,A,f,g)$ be a lingo. Then, $\forall d_1, d'_1 \in D_1$, $\forall a \in A$ $f(d_1, a) = f(d'_1, a) \Rightarrow d_1 = d'_1$. 
\end{lemma}

\noindent {\bf Proof}:
Applying $g$ on both sides of the condition, by the
definition of lingo we get,
$d_1 = g(f(d_1,a),a) = g(f(d'_1,a),a) = d'_1$. $\Box$

\begin{lemma}\label{lem:lemma2}
Let $\Lambda=(D_1,D_2,A,f,g)$ be a lingo. Then, $\forall d_1 \in D_1$, $\forall d_2 \in D_2$, $\forall a \in A$, $d_2 = f(d_1,a) \Rightarrow d_1 = g(d_2,a)$.
\end{lemma}

\noindent {\bf Proof}:  Applying $g$ to both sides of the
condition, by the
definition of lingo we get,
$g(d_2,a) = g(f(d_1,a),a) = d_1$.  $\Box$.

\section{Maude Specifications of Lingos and Dialects }\label{app:examples}

In this section, we provide more details
on how lingos and dialects are specified as formal executable specifications in Maude.
Some details are omitted for the sake of a simpler exposition.

\subsection{Theories, Views and Parameterised Modules in Maude}\label{app:param}

Maude \cite{maude-book} supports \emph{parameterized functional modules}  having a \emph{free algebra semantics}.
They are parameterized by one of more functional theories.  For example, \verb+fmod MSET{X :: TRIV} endfm+
is a parameterized module with parameter theory \texttt{TRIV} having a single parameter type \texttt{Elt} and no
operations or equations that map any chosen set $X$ of elements to the free algebra of (finite) multisets on the elements
of $X$.  The choice of an actual parameter $X$ is made by a Maude \emph{view} instantiating in this case 
\texttt{TRIV} to the chosen data type $X$.  For example, we can choose as elements the natural numbers
by a view \texttt{Nat} mapping \texttt{TRIV} to the functional module \texttt{NAT} of natural numbers
and mapping the parameter type (called a sort) \texttt{Elt} to the type (sort) \texttt{Nat} in \texttt{NAT}.
The semantics of the thus instantiated module \verb+fmod MSET{Nat}+ is the free algebra
of multisets with elements in the natural numbers. 
Likewise, \emph{parameterized system modules}  have a \emph{free model semantics}.
For example, a parameterized system module \verb+mod CHOICE{X :: TRIV} endm+ that imports 
the parameterized functional module \verb+fmod MSET{X :: TRIV} endfm+ can have a  rewrite rule
for non-deterministically choosing an element
in a multiset of elements.  Then, the instantiation \verb+mod CHOICE{Nat} endm+ provides non-deterministic
choice of an element in a multiset of natural numbers.  

The class of all \emph{lingos} is specified by a \emph{functional theory} \texttt{LINGO}, whereas
specific lingos are specified as \emph{functional modules} (which can be parameterised) with initial algebra semantics, i.e.,
they are executable algebraic data types.  Instead, the \emph{dialect} formal pattern is
specified as a \emph{parameterized system module} with two parameters: a lingo parameter theory and a protocol parameter theory.  
When instantiated by views to:
(i) a concrete lingo $\Lambda$, and (ii) a concrete protocol $\mathcal{P}$, that is, a generalized actor rewrite theory,
application of the dialect pattern
results in a new protocol that endows the original protocol  $\mathcal{P}$ with the dialect behavior ``speaking'' 
the chosen lingo $\Lambda$.

\subsection{Lingos}\label{app:ex-lingos}
In Maude, an equational theory $(\Sigma,E)$ can be specified as a \emph{functional theory}.  Therefore, the notion of lingo given in Definition~\ref{def:lingo} has the following natural specification in Maude:

\begin{Verbatim}
fth LINGO is
    sorts D1 D2 A . op f : D1 A -> D2 . op g : D2 A -> D1 .
    var d1 : D1 . var a : A . eq g(f(d1,a),a) = d1 .
endfth
\end{Verbatim}

Since lingos are a key component of dialects\cite{ProtocolDialects-FormalPatterns}, there is some extra information that a lingo must provide so a dialect can apply it. We easily extend the above theory, by using Maude's \texttt{including} on \texttt{LINGO}. For more information on importations in Maude (e.g., \texttt{protecting}), see \cite{maude-book}. 
Theory \texttt{PMLINGO} introduces: i) a function called \texttt{param} for computing the corresponding parameter value $\mathit{param}(n) \in A$ from a natural number $n$, and ii) ingress and egress arities, specifying how many messages the lingo receives as inputs, and how many messages does it generate as outputs.

\begin{Verbatim}
fth PMLINGO is
    protecting NAT . including LINGO .
    op param : Nat -> A .
    ops ingressArity egressArity : -> NzNat .
endfth
\end{Verbatim}

Let us specify in Maude the parametrised data type $\{0,1\}^{n}$ parametric on a non-zero bit-vector length $n$ and its associated \texttt{xor} function.  Since in Maude's built-in module \texttt{NAT} of (arbitrary precision) natural numbers there is a function \texttt{xor} that, given numbers $n,m$, computes the number  $n \; \texttt{xor}\; m$ whose associated bit-string is the bitwise exclusive or of $n$ and $m$ when represented as bit-strings, we can define such a parametrised module as follows:

\begin{Verbatim}
fth NzNATn is 
    protecting NAT . 
    op n : -> NzNat [pconst] . 
endfth

fmod BIT-VEC{N :: NzNATn} is protecting NAT .
    sorts BitVec{N} BitStr{N} . subsort BitVec{N} < BitStr{N} .

    op [_] : Nat -> BitStr{N} [ctor] . vars N M : Nat .
    cmb [N] : BitVec{N} if N < (2 ^ N$n) .

    op _xor_ : BitStr{N} BitStr{N} -> BitStr{N} [assoc comm] .
    eq [N] xor [M] = [N xor M] .
endfm

view 8 from NzNATn to NAT is 
    op n to term 8 . 
endv
\end{Verbatim}

We can easily define its instantiation for bytes.

\begin{Verbatim}
fmod BYTE is
 protecting BIT-VEC{8} * (sort BitVec{8} to Byte) .
endfm
\end{Verbatim}

And now we can evaluate some expressions.

\begin{Verbatim}
reduce in BYTE : [3] xor [5] .
result Byte: [6]
\end{Verbatim}

After having defined \verb@BIT-VEC{N :: NzNATn}@, 
let us now define the parametrized lingo $\Lambda_{\mathit{xor}}=(\{0,1\}^{n},\{0,1\}^{n},\{0,1\}^{n},\oplus,\oplus)$ from Example~\ref{ex:xor-lingo} in Maude as follows:
    
\begin{Verbatim}
fmod IDLINGO{L :: LINGO} is
    op f' : L$D1 L$A -> L$D2 . op g' : L$D2 L$A -> L$D1 .
    var d1 : L$D1 .  var d2 : L$D2 .  var a : L$A .
    eq f'(d1,a) = f(d1,a) . eq g'(d2,a) = g(d2,a) .
endfm

view xorl{N :: NzNATn} from LINGO to BIT-VEC{N} is
    sort D1 to BitVec{N} . sort D2 to BitVec{N} . sort A to BitVec{N} .
    op f to _xor_ . op g to _xor_ .
endv

fmod XOR-L{N :: NzNATn} is
    protecting IDLINGO{xorl{N}} . 
endfm
\end{Verbatim}

Mathematically, what the above view \texttt{xorl} states is that for each value of the parameter \texttt{n}, choosing \texttt{f} and \texttt{g}
to be the \verb@_xor_@ function is an executable algebra $\Lambda_{\mathit{xor}}$ satisfying the equation in the \texttt{LINGO} theory.

Lingos are used by Dialects to transform the messages from, and to, the underlying object. 
The meaning of the parametrized module \verb@IDLINGO{L :: LINGO}@
is that of the \emph{identical lingo} obtained
by instantiating the theory \texttt{LINGO} with concrete data
types \texttt{D1}, \texttt{D2} and \texttt{A}, and concrete
functions \texttt{f} and \texttt{g}.  Recall that such instantiations
are achieved in Maude by means of a \emph{view}.  The only notational
difference is that we rename \texttt{f} and \texttt{g} by
\texttt{f'} and \texttt{g'}.  We then instantiate the theory
\texttt{LINGO} to   \verb@BIT-VEC{N}@ by means of a \emph{parametrized view}
\verb@xorl{N :: NZNATn}@.  Our desired parametrized lingo
$\Lambda_{\mathit{xor}}=(\{0,1\}^{n},\{0,1\}^{n},\{0,1\}^{n},\oplus,\oplus)$
is then the parametrized module \verb@XOR-LINGO{N :: NZNATn}@. 

We can then instantiate \verb@XOR-LINGO{N :: NZNATn}@ to bytes as follows:

\begin{Verbatim}
fmod BYTE-L is
    protecting XOR-L{8} * (sort BitVec{8} to Byte) .
endfm
\end{Verbatim}

And now we can invoke the lingo function \texttt{f'}.

\begin{Verbatim}
reduce in BYTE-L : f'([3], [5]) .
result Byte: [6]
\end{Verbatim}

As we said, Maude's built in module \texttt{NAT} for natural numbers has a function \texttt{xor}. This allows us to easily define lingo $\Lambda_{xor.BSeq}$ from Example~\ref{ex:XOR-BSeq-Lingo} as follows:

\begin{Verbatim}
view xorl-seq-l from LINGO to NAT is
    sort D1 to Nat . sort D2 to Nat . sort A to Nat .
    op f to _xor_ . op g to _xor_ .
endv

fmod XOR-BSEQ-L is 
    protecting IDLINGO{xorl-seq-l} . 
endfm
\end{Verbatim}

And evaluate some expressions.

\begin{Verbatim}
reduce in XOR-BSEQ-L : f'(3, 5) .
result NzNat: 6
==========================================
reduce in XOR-BSEQ-L : g'(6, 5) .
result NzNat: 3
\end{Verbatim}

In a very similar way, the functional module \texttt{NAT} also provides us with two necessary operations for the divide and check lingo of Example~\ref{ex:d&c-lingo}. These operations are the quotient (\texttt{quo}) and the reminder (\texttt{rem}) functions. Furthermore, we also need a structure for pairs of numbers to represent the output $D_2$ of such a lingo. Such need is satisfied by a functional module \texttt{NAT-PAIR} declaring and defining the pair structure and its two projection operations. Thus, lingo $\Lambda_{D\&C}$ can be formally specified in Maude as the functional module \verb@D&C-L@
as follows\footnote{Notice how \texttt{NAT-PAIR} includes \texttt{NAT}, allowing us to define the lingo in a straightforward way.}:

\begin{Verbatim}
fmod NAT-PAIR is protecting NAT .
    sort NatPair .
    op [_,_] : Nat Nat -> NatPair [ctor] . ops p1 p2 : NatPair -> Nat .
    vars n m : Nat . eq p1([n,m]) = n . eq p2([n,m]) = m .
endfm

view D&C-ling from LINGO to NAT-PAIR is
    sort D1 to Nat . sort D2 to NatPair . sort A to Nat .
    var n : D1 . var P : D2 . var a : A .
    op f(n,a) to term [(n + (a + 2)) quo (a + 2), 
                       (n + (a + 2)) rem (a + 2)] .
    op g(P,a) to term sd(((p1(P) * (a + 2)) + p2(P)), (a + 2)) .
endv

fmod D&C-L is
    protecting IDLINGO{D&C-ling} .
endfm
\end{Verbatim}

And evaluate some expressions.

\begin{Verbatim}
reduce in D&C-L : f'(13, 3) .
result NatPair: [3, 3]
==========================================
reduce in D&C-L : g'([3, 3], 3) .
result NzNat: 13
==========================================
reduce in D&C-L : g'(f'(13, 3), 3) .
result NzNat: 13
\end{Verbatim}


\subsection{Horizontal Composition}\label{app:ex-lingo-horizontal-comp}


The horizontal composition operation, presented in Definition~\ref{def:hor-comp}, of $k$ lingos sharing a common input data type can be specified in Maude as a functional module parametrised by the $k$-\texttt{LINGO} parameter theory.  We illustrate 
here the construction for $k=2$, where $2$-\texttt{LINGO} is abbreviated to \texttt{DILINGO}.  
Furthermore, we extend the notion of \texttt{DILINGO} with theory \texttt{W-DILINGO}, which adds two constant values to represent weights. These weights will be used for the bias when choosing a lingo, following Definition~\ref{def:hor-comp}.

\begin{Verbatim}
fth DILINGO is
    protecting NAT .
    sorts D1 D21 D22 A1 A2 .
    
    op f1 : D1 A1 -> D21 . op f2 : D1 A2 -> D22 .
    op g1 : D21 A1 -> D1 . op g2 : D22 A2 -> D1 .
    op d0.1 : -> D21 . op d0.2 : -> D22 .

    var a1 : A1 . var a2 : A2 . var d1 : D1 .
    eq g1(f1(d1, a1), a1) = d1 .
    eq g2(f2(d1, a2), a2) = d1 .

    ops ingressArity egressArity : -> NzNat .

    op param1 : Nat -> A1 .
    op param2 : Nat -> A2 .
endfth

fth W-DILINGO is
    including DILINGO .
    
    ops w1 w2 : -> NzNat [pconst] .
endfth
\end{Verbatim}

Next, we can use this theory in a parameterised functional module \texttt{HOR-COMP} to specify in Maude the horizontal composition of two lingos. Notice that the imported module \texttt{FDICE} (omitted) models a die biased with weights, which is used by the pseudo-random function \texttt{throw}  as part of the lingo-election procedure.

\begin{Verbatim}
fmod HOR-COMP{WL :: W-DILINGO} is
    protecting FDICE .
    sorts D2 A . subsorts WL$D21 WL$D22 < D2 .

    var N : Nat .
    op param : Nat -> Nat .
    eq param(N) = N .

    op +f : WL$D1 Nat -> D2 .
    op +g : D2 Nat -> WL$D1 .

    var d1 : WL$D1 . var d2 : D2 .

    ceq +f(d1, N) = f1(d1, param1(N)) if throw(N, WL$w1 <;> WL$w2) == 1 .
    ceq +f(d1, N) = f2(d1, param2(N)) if throw(N, WL$w1 <;> WL$w2) == 2 .

    ceq +g(d2, N) = g1(d2, param1(N)) if throw(N, WL$w1 <;> WL$w2) == 1 .
    ceq +g(d2, N) = g2(d2, param2(N)) if throw(N, WL$w1 <;> WL$w2) == 2 .
endfm
\end{Verbatim}

We only need to specify a view from \texttt{DILINGO} to \texttt{NAT-PAIR} and \textit{feed} it to functional module \texttt{HOR-COMP} to obtain a horizontal composition of lingos. A view for Example~\ref{ex:horizontal-comp}, and its instantiation in the functional module \texttt{XORBS\&DC-L}, showing the horizontal composition of $\Lambda_{xor.BSeq}$ and $\Lambda_{D\&C}$ lingos, is\footnote{It is important to note that, in this view, we have decided to give more weight to the second lingo, thus obtaining a \textit{biased} die. If we want a \textit{fair} die, we should give the same weights to both lingos.}:

\begin{Verbatim}
view xor&dc from W-DILINGO to NAT-PAIR is 
    sort D1 to Nat . sort D21 to Nat . sort D22 to NatPair .
    sort A1 to Nat . sort A2 to Nat .

    op f1 to _xor_ . op g1 to _xor_ .
    op f2 to term [(n+(a+2)) quo (a+2), (n+(a+2)) rem (a+2)] .
    op g2 to term sd(((p1(P) * (a+2)) + p2(P)), (a+2)) .

    op d0.1 to term 0 . op d0.2 to term [0,0] .

    op w1 to term 1 .
    op w2 to term 5 .
endv

fmod XORBS&DC-L is
    protecting HOR-COMP{xor&dc} .
endfm
\end{Verbatim}

Additionally, since the horizontal composition of two lingos is itself a lingo, we can specify the following view from theory \texttt{PMLINGO} to our recently defined \texttt{HOR-COMP\{DL :: DILINGO\}} module. Such a view can then be used to instantiate a dialect with an horizontal composition as a lingo, as we show in Appendix~\ref{app:ex-protocols}.

\begin{Verbatim}
view HOR-COMP{WL :: W-DILINGO} from PMLINGO to HOR-COMP{WL} is
    sort D1 to WL$D1 .
    sort D2 to D2 .
    sort A to Nat .

    op f to +f .
    op g to +g .
endv
\end{Verbatim}

\subsection{Functional composition}\label{app:ex-lingo-func-comp}

The  functional composition operation $(\Lambda,\Lambda') \mapsto \Lambda \odot \Lambda'$, presented in Definition~\ref{def:fun-comp},
can be formally specified in Maude
with a parametrised module \texttt{FUN-COMP}, parametric
on the 
theory \texttt{COMP-LINGOS} 
of (functionally) composable
lingos. 

\begin{Verbatim}
fth COMP-LINGOS is
    sorts D1 D2 D3 A A'.
    op f : D1 A -> D2 . op f' : D2 A' -> D3 .
    op g : D2 A -> D1 . op g' : D3 A' -> D2 .
    
    var d1 : D1 . var d2 : D2 . var a : A . var a' : A' .
    eq g(f(d1,a),a) = d1 . eq g'(f'(d2,a')) = d2 .

    op param : Nat -> A . op param' : Nat -> A' .
endfth

fmod FUN-COMP{CL :: COMP-LINGOS} is
    sort pairA . op a[_,_] : CL$A CL$A' -> pairA [ctor] .
    
    var d1 : D1 . var d3 : D3 . var a : A . var a' : A' . var n : Nat .
    op f.f' : D1 pairA -> D3 . eq f.f'(d1,a[a,a']) = f'(f(d1,a),a') .
    op g*g' : D3 pairA -> D1 . eq g*g'(d3,a[a,a']) = g(g'(d3,a'),a) .

    op param.param' : Nat -> pairA .
    eq param.param'(n) = a[param(n),param'(n)] .
endfm
\end{Verbatim}

Similar to horizontal composition, we can now create a view from \texttt{NAT-PAIR} to \texttt{COMP-LINGOS} and \textit{feed} it to the parametrised functional module \texttt{FUN-COMP} to obtain a functional composition. A view for Example~\ref{ex:functional-comp}, and its instantiation in functional module \texttt{XOR*D\&C-L} that functionally composes $\Lambda_{xor.BSeq}$ and $\Lambda_{D\&C}$, is as follows:

\begin{Verbatim}
view xor*dc from COMP-LINGOS to NAT-PAIR is
    sort D1 to Nat .
    sort D2 to Nat .
    sort D3 to NatPair .
    sort A to Nat .
    sort A' to Nat .
    
    op f to _xor_ . op g to _xor_ .
    op f' to term [(n+(a+2)) quo (a+2), (n+(a+2)) rem (a+2)] .
    op g' to term sd(((p1(P) * (a+2)) + p2(P)), (a+2)) .

    op param to random . op param' to random .
endv

fmod XOR*D&C-L is
    is protecting FUN-COMP{xor*dc} .
endfm
\end{Verbatim}

Furthermore, since the functional composition of two lingos yields a lingo, we can specify a view from the functional theory of \texttt{PMLINGO} to our functional composition module \texttt{FUN-COMP\{CL\}}. This view allows us to instantiate a dialect with a functional composition as a lingo in Appendix~\ref{app:ex-protocols}.

\begin{Verbatim}
view FUN-COMP{CL :: COMP-LINGOS} from PMLINGO to FUN-COMP{CL} is
    sort D1 to CL$D1 . sort D2 to CL$D3 . sort A to pairA .
    op f to f.f' . op g to g*g' . op param to param.param' .
endv
\end{Verbatim}

\subsection{Dialects}\label{app:ex-dialects}

Dialects are protocol transformations.  We first need to explain how protocols are specified in Maude. Maude supports special syntax for
rewrite theories that are concurrent object systems and, in particular, for protocols, which are generalized actor rewrite theories.
Maude theories are in this way specialized to \emph{object-oriented theories}, with syntax \texttt{oth}.  Likewise, Maude modules
are in this way specialized to \emph{object-oriented modules}, with syntax \texttt{omod}. 
Protocols will be specified in Maude as object-oriented modules.
We are interested in the class of all protocols, i.e., of all
generalized actor rewrite theories; but for each protocol we wish to make explicit the data type of its \emph{payloads}.
This can be achieved with the following \texttt{PROTOCOL} object theory, where we require a data type (sort) for the payload, and 
specify a message format to carry that payload. 
\begin{Verbatim}
oth PROTOCOL is
    sort Payload .
    msg to_from_:_ : Oid Oid Payload -> Msg .
endoth
\end{Verbatim}

Therefore, a concrete protocol and its payload data type can be identified by a \emph{view} from the \texttt{PROTOCOL} object-oriented theory
to the Maude object-oriented module specifying the protocol itself.  With the \texttt{PROTOCOL} object-oriented theory
and the lingo functional theory given in Section~\ref{app:ex-lingos}, following the description of Section~\ref{Dialects-sub} 
we can specify the \emph{dialect formal pattern} in Maude
by means of the following parameterised object module (details omitted):

\begin{Verbatim}
omod DIALECT{L :: LINGO, P :: PROTOCOL} is
    ...
    class Dialect{L,P} | conf : Configuration, 
                         in-buffer : Configuration, 
                         peer-counters : Map{Oid, Nat} .
    ...
endom
\end{Verbatim}

The above \verb@DIALECT{L :: LINGO, P :: PROTOCOL}@ parameterised object
module also contain the \emph{rules} depicted in Figure~\ref{fig:wrapper-actor-rules}, as well as various other sorts and auxiliary functions. 

\subsection{Various Dialect Examples based on the MQTT Protocol}\label{app:ex-protocols}

MQTT~\cite{HunkelerTS08} is a lightweight, publish-subscribe, protocol 
popular in devices with limited resources or network bandwidth. MQTT provides no security itself, and, although it can be run over secure transport
protocols such as TLS, it is known to be commonly misconfigured [29]. This
makes it an attractive application for protocol dialects [1].

We formally specified in Maude the semantics of MQTT as a parameterized object-oriented module \texttt{MQTT\{D\}}, parametric on a data type \texttt{D} of payloads (details are omitted for brevity). As an example, we show the process of a client connecting to a broker with the following rules.

\begin{Verbatim}
rl [mqtt/C/send-connect] :
    < Me : MqttClient | peer : nothing, cmdList : connect(B); CMDL, 
                        Atrs >
    =>
    < Me : MqttClient | peer : nothing, cmdList : CMDL, Atrs > 
    (to B from Me : connect(B)) .

rl [mqtt/B/accept-connect] : 
    < Me : MqttBroker | peers : Ps, Atrs > (to Me from O : connect(Me)) 
    =>
    < Me : MqttBroker | peers : insert(O, Ps), Atrs > 
    (to O from Me : connack) .

rl [mqtt/C/recv-connack] : 
    < Me : MqttClient | peer : nothing, Atrs > (to Me from O : connack)
    =>
    < Me : MqttClient | peer : O, Atrs > .
\end{Verbatim}

The first rule models the behaviour of an \texttt{MqttClient} with identifier \texttt{Me} sending a \texttt{connect} message to an \texttt{MqttBroker} with identifier \texttt{B}. The second rule represents what the \texttt{MqttBroker} will do upon receiving the \texttt{connect} message, i.e. add the \texttt{MqttClient} to it's set of peers and send back to the sender an acknowledgement message. With the third rule, the \texttt{MqttClient} processes any incoming acknowledgement message by setting its \texttt{peer} attribute with the senders identifier.

We can make explicit that the parameterised functional moule \texttt{MQTT\{D\}} is a protocol, by defining the following view from the object-oriented theory \texttt{PROTOCOL} to it:

\begin{Verbatim}
view MqttProtocol{D :: TRIV} from PROTOCOL to MQTT{D} is
    sort Payload to D$Elt .
endv
\end{Verbatim}

Finally, we can use this view, together with the provided views for lingos and lingo transformations in Sections~\ref{app:ex-lingos}, \ref{app:ex-lingo-horizontal-comp} and \ref{app:ex-lingo-func-comp}, to instantiate the corresponding dialects mentioned in Example~\ref{ex:coealescing}:

\begin{itemize}
    \item $\mathcal{D}_{\Lambda_{XOR}\{8\}}(Mqtt\{BitSEQ\{8\}\})$ can be obtained with \texttt{DIALECT\{xorl\{8\},\\MQTT\{BitVEC\{8\}\}\}}, where \texttt{BitVEC\{8\}} is a parametric view to represent bit-sequences of length 8 in this case.
    \item $\mathcal{D}_{\Lambda_{XOR}\{BitSEQ\{256\}\}}(Mqtt\{BitSEQ\{256\}\})$ can be obtained with \texttt{\\DIALECT\{xorl\{256\},MQTT\{BitVEC\{256\}\}\}}.
    \item $\mathcal{D}_{\Lambda_{xor.BSeq}}(Mqtt\{Nat\})$ can be obtained with \texttt{DIALECT\{xor-seq-l,\\MQTT\{Nat\}\}}, where \texttt{Nat} is a \textit{trivial} view mapping elements to the natural numbers.
    \item $\mathcal{D}_{\Lambda_{D\&C}}(Mqtt\{Nat\})$ can be obtained with \texttt{DIALECT\{D\&C-ling,MQTT\{Nat\}}.
    \item $\mathcal{D}_{\Lambda_{xor.BSeq} \oplus_{\overrightarrow{d}_{0}} \Lambda_{D\&C}}(Mqtt\{Nat\})$ can be obtained with \texttt{DIALECT\{HOR-COMP\{xor\&dc\},\\MQTT\{Nat\}}.
    \item $\mathcal{D}_{\Lambda_{xor.BSeq} \odot \Lambda_{D\&C}}(Mqtt\{Nat\})$ can be obtained with \texttt{DIALECT\{FUN-COMP\{xor*dc\},\\MQTT\{Nat\}}.
\end{itemize}

\section{Authenticating Lingos}\label{app:authenticating-lingos}

\noindent In this section we make use of the pseudo-random numbers that are  used to form the parameters shared by two  principals (Alice and Bob)  to transform a lingo into a lightweight \emph{authenticating lingo}. Although we do not assume that these numbers shared by Alice and Bob are unknown to other principals in the enclave, our trust assumptions imply that those other principals will not take advantage of that knowledge.



The key idea about an authenticating lingo is that, if Bob receives
a transformed payload $f(d_1,a)$ supposedly from Alice, he has a way
to check that it was sent by Alice because he can extract from
$f(d_1,a)$ some data that was produced using
information shared between Alice and Bob that is unavailable to an outside attacker.
In this case it will be the result of computing a one-way hash function $\mathit{hash}$ over their identities 
and their current shared parameter $a$. To make this possible, we need to make explicit two
data types involved in the communication.  First, honest participants
in the communication protocol using lingo $\Lambda$ typically have 
corresponding \emph{object identifiers} that uniquely identify 
them and
belong to a finite data type $\mathit{Oid}$.
Second, the result of computing the $\mathit{hash}$  belongs to another finite data type $H$.  Here is
the definition:

\begin{definition} An \emph{authenticating lingo} is a tuple 
$((D_1,\allowbreak D_2,\allowbreak  A,\allowbreak f,\allowbreak g),\allowbreak \mathit{Oid},\allowbreak H,\allowbreak \mathit{param},\allowbreak 
\mathit{hash},\allowbreak \mathit{code})$
such that:
\begin{enumerate}
\item $(D_1 , D_2 ,A,f,g)$ is a lingo;

\item $\mathit{Oid}$ and $H$ are finite data types of, respectively,
object identifiers and outputs of $\mathit{hash}$.

\item $\mathit{param}$ is a function $\mathit{param}: \mathbb{N} \times  \mathit{Oid} \otimes \mathit{Oid} \rightarrow A$;

\item $\mathit{hash}$ is a function $\mathit{hash}: \mathbb{N} \times  \mathit{Oid} \otimes \mathit{Oid} \rightarrow H$;\footnote{$\mathit{hash}$ should have good
properties as a hash function, that is, it should be noninvertible and collision-resistant.  This seems impossible, since $H$ is finite
but $\mathbb{N}$ is not.  But this problem is only apparent.  In practice the
values $n \in \mathbb{N}$ used as first arguments of $\mathit{hash}$
(and also as first arguments of $\mathit{param}$)
will be random numbers $n$ such that $n \in \mathbb{N}_{< 2^{k}}$
for some fixed $k$.} and

\item $\mathit{code}$ is a function $\mathit{code}: D_2 \times A \rightarrow H$
\end{enumerate}
such that $\forall d_1 \in D_1$, $\forall n \in \mathbb{N}$, $\forall (\mathbf{A},\mathbf{B}) \in \mathit{Oid} \otimes \mathit{Oid}$,
\[\mathit{code}(f(d_1,\mathit{param}(n,(\mathbf{A},\mathbf{B}),\mathit{param}(n,(\mathbf{A},\mathbf{B}))) = 
\mathit{hash}(n,(\mathbf{A},\mathbf{B})).\]
\end{definition}

\begin{example} \label{ex:authlingo}
Let us illustrate this notion by means of a transformation
$\Lambda \mapsto \Lambda^{\alpha}$ that generates an authenticating
lingo $\Lambda^{\alpha}$ from a lingo $\Lambda$.  For concreteness,
assume that $\Lambda$ is of the form $\Lambda = (\mathbb{N}_{< 2^{n}},
\mathbb{N}_{< 2^{m}},A_{0},f_{0},g_{0})$ with $A_{0}$ a finite set,
and that there is a function $\mathit{param}_{0}: \mathbb{N}_{\leq 2^{k}}
\rightarrow A_{0}$ assigning to each random number $n \in \mathbb{N}_{< 2^{k}}$
a corresponding parameter in $A_{0}$.  This reflects the fact that, in practice,
two honest participants will share a common parameter $a_{0} \in A_{0}$
by sharing a common secret random number $n$ and using $\mathit{param}_{0}$
to generate $a_0$.
Assuming a finite data type $\mathit{Oid}$ of honest participants, 
choose a function $\mathit{hash}: \mathbb{N}_{< 2^{k}} \times  \mathit{Oid} \otimes \mathit{Oid} \rightarrow \mathbb{N}_{< 2^{j}}$,
with $j$ and $k$ big enough for $\mathit{hash}$
to be computationally infeasible to invert.
Choose a function $\iota: \mathbb{N}_{< 2^{k}} \rightarrow \mathit{Invo}(m+j)$ mapping each random number $n$ to an \emph{involution}\footnote{An involution
is a permutation $\sigma$ such that $\sigma$ is its own inverse permutation, i.e.,
$\sigma = \sigma^{-1}$.} 
$\sigma$
on $m+j$ elements. 
By abuse of notation, if $\overrightarrow{b}= b_1 , \ldots , b_{m + j}$ is a bit sequence of length $m + j$,
we let $\sigma(b_1 , \ldots , b_{m+j}) = b_{\sigma(1)} , \ldots , b_{\sigma(m+j)}$.

\noindent Then define $\Lambda^{\alpha}$ as follows:
$
\Lambda^{\alpha} = ((\mathbb{N}_{< 2^{n}},
\mathbb{N}_{< 2^{m+j}},A_{0} \times \mathit{Invo}(m+j) \times \mathbb{N}_{< 2^{j}} ,f,g), \\
\mathit{Oid},\mathbb{N}_{< 2^{j}},\mathit{param},\mathit{hash},\mathit{code})$
where:
\begin{enumerate}
\item $f(d_1 ,(a_{0}, \sigma,d)) = \sigma(f_{0}(d_1 ,a_0 );d)$,
where we view $f_{0}(d_1 ,a_0 )$ and $d$ as bit vectors of
respective lengths $m$ and $j$, $f_{0}(d_1 ,a_0 );d$ denotes
their concatenation as bit sequences. For protocol participants $\mathbf{A}$ and $\mathbf{B}$,
$d$ will be of the form $d = \mathit{hash}(n,(\mathbf{A},\mathbf{B}))$ (see (3) below).


\item $g(\overrightarrow{b},(a_{0}, \sigma,d))=g_{0}(p_{1}(\sigma(\overrightarrow{b})),a_{0})$,
where by abuse of notation identify $\overrightarrow{b}$ with a number smaller than $2^{m+j}$, and where $p_1$
(resp. $p_{2}$)
projects an $m+j$ bitvector to its first $m$ bits (resp. its last $j$ bits).

\item $\mathit{param}: \mathbb{N}_{< 2^{k}} \times  \mathit{Oid} \otimes \mathit{Oid} \rightarrow A_{0} \times \mathit{Invo}(m+j) \times \mathbb{N}_{< 2^{j}}$ is the function
$\lambda (n,(\mathbf{A},\mathbf{B}))\; (param_{0}(n),\iota(n),\mathit{hash}(n,(\mathbf{A},\mathbf{B}))$.

\item $\mathit{code}(\overrightarrow{b},(a_{0}, \sigma,d))= p_{2}(\sigma(\overrightarrow{b}))$.
Again, for protocol participants $\mathbf{A}$ and $\mathbf{B}$,
$d$ will be of the form $d = \mathit{hash}(n,(\mathbf{A},\mathbf{B}))$.
\end{enumerate}
\end{example}

\noindent The reader can check that $\Lambda^{\alpha}$ is an authenticating 
lingo.

\vspace{1.5ex}

We now give an informal description of the reasoning behind our conjecture that the transformation in Example \ref{ex:authlingo} can turn some malleable lingos into non-malleable ones. Suppose that the transformation is applied to the  lingo  $\Lambda_{\mathit{xor}}^{\sharp}$, whose malleability is shown in  Example \ref{ex:xor-sharp-mall}. In this lingo,  $f$ is defined by $ f(d_1,a) = \sigma (b_0 \oplus d_1 ; b_1 \oplus d_1 ; d)$, where the secret parameter $a$ is $b_0 ; b_1 ; \sigma$.  Note that, although the attacker sees $f(d_1,a)$, it has no ability to learn anything about $\sigma$, because from its point of view $\sigma (b_0 \oplus d_1 ; b_1 \oplus d_1 ; d)$ is random.
 
 We now suppose that the attacker applies some function to $ f(d_1,a) = \sigma (b_0 \oplus d_1 ; b_1 \oplus d_1 ; d)$ in order to produce another dialected message.  This has the effect of flipping some bits in $ f(d_1,a)$ and leaving the other bits unchanged.  That means that we can analyze an attacker’s behavior in terms of which bits it flips, without referring to the function it applies.  This makes the result of the attack easier to analyze combinatorially, at least for this example.    
 
 Note that, if  bits are flipped in $b_0 \oplus d_1 ; b_1 \oplus d_1 ; d$, the resulting string  can still pass the $f$-check if and only if  (i) no bit is flipped in $d$, and (ii) if any bit is flipped in $b_1 \oplus d_1$, then so is the bit  in $b_2 \oplus d_1$, and vice versa. We will call these constraints the  f-check constraints.  That means that the positions of the bits flipped in $\sigma (b_1 \oplus d_1 ; b_2 \oplus d_1 ; d)$ must be mapped by $\sigma$ to  a set of positions satisfying the f-check constraints.   Thus,  if we let $d'_1$ be the result of flipping bits in $b_0 \oplus d_1 ; b_1 \oplus d_1 ; d$, then $d'_1$ will be accepted by the recipient if and only if $\sigma$ maps the positions of the flipped bits to a set of positions satisfying the $f$-check constraints. We now analyze the asymptotic behavior of the probability of this occurring.
 
We set up the problem with the following parameters.  Let $s$ be a bitstring of length $2n + k$, formed by concatenating three bitstrings: (i) $c_1$ with  positions $1, \ldots, n$, (ii) $c_2$ with positions $n+1, \ldots, 2n$, and (iii)
$c_3$ with  positions $2n+1, \ldots, 2n+k$.  
 Let $Q$ be a set of $2m$ positions from $\{1, \ldots, 2n+k\}$.   Let $T = \{i_1, \ldots, i_m,\, n + i_1, \ldots, n + i_m\}$, where $i_j \in \{1, \ldots, n\}$ for each $j$, 
 so that $T$ consists of $m$ positions from $c_1$ and their corresponding mirror positions in $c_2$. Thus $|T|  = |Q| = 2m$.  
 We wish to compute the probability that $\sigma$ bijects $Q$ onto $T$.  
 We first considered the case in which the attacker wishes to flip the bits of a  specific set $T$, and analyzed the asymptotic behavior. 
 In this case the probability was easy to compute and it was also straightforward to show that it converges to 0 as $n \rightarrow \infty$ and $k$ is held constant.
  
  However, things got more complex when we considered the case in which the attacker's goal  is to pick a $Q$ such that there was \emph{some} $T$ satisfying the $f$-checkability conditions such that $\sigma(Q) = T$.   In this case, convergence to 0 was not guaranteed for constant $k$, but  investigations showed that it is possible to guarantee convergence to 0 by allowing  $n$ and $k$ to approach infinity, with $k$ increasing linearly as $n$, thus giving the defender a way  to defeat the attacker. \footnote{We are grateful to the Claude Sonnet 4.6 AI system both for the derivation of these probabilities and the analysis of their behavior.}

\section{Lingo Transformations as Formal Patterns}\label{app:FormPatts}

In a declarative programming language like Maude, a program is a \emph{theory} $T$ in a computational logic, and
a \emph{program transformation}, say $\tau$, is a mapping $\tau : T \mapsto \tau(T)$ that transforms program $T$ into program $\tau(T)$.
Such program transformations can be fruitfully understood as \emph{formal patterns} \cite{DBLP:conf/adt/Meseguer22}, that is, as \emph{formally-defined solutions
to specific software construction problems}. In  general, such transformations may transform not a single theory but several, and
may depend on additional parameters.  They are partial functions of the general form:
\[
\tau : \mathcal{C}_{1}\times \ldots \mathcal{C}_{n}\times \mathit{Params} \ni (T_{1},\ldots,T_{n},\vec{p}) \mapsto \tau(T_{1},\ldots,T_{n},\vec{p}) \in \mathcal{C}
\]
where $\mathcal{C}_{1},\ldots, \mathcal{C}_{n}$ are \emph{theory classes} specifying the \emph{assumptions} about the input theories $T_{1},\ldots,T_{n}$,
and $\mathcal{C}$ is the theory class to which $\tau(T_{1},\ldots,T_{n},\vec{p})$ is \emph{guaranteed} to belong.  We can illustrate formal patterns with several
examples from Section \ref{sec:lingos}.  For example, the $\Lambda \mapsto \Lambda^{\sharp}$ transformation
is a formal pattern where $\mathcal{C}_{1}$ is the class of equational theories $(\Sigma,E)$ specifying lingos as algebraic data types $\mathbb{T}_{\Sigma/E}$,
and $\mathcal{C}$ is the class of equational theories $(\Sigma,E)$ specifying $f$-\emph{checkable} lingos, to which $\Lambda^{\sharp}$ belongs.
Likewise, the transformations $\overrightarrow{\Lambda} \mapsto \bigoplus_{\overrightarrow{d_{0}}} 
\overrightarrow{\Lambda}$ (resp. $(\Lambda,\Lambda') \mapsto \Lambda \odot \Lambda'$) are formal patterns with
$\mathcal{C}_{1}=\ldots=\mathcal{C}_{n}$  (resp. $\mathcal{C}_{1}=\mathcal{C}_{2}$)
the class of equational theories $(\Sigma,E)$ specifying lingos, and so is
$\mathcal{C}$.  Note that the $\overrightarrow{\Lambda} \mapsto \bigoplus_{\overrightarrow{d_{0}}} 
\overrightarrow{\Lambda}$ transformation depends on the parametric vector $\overrightarrow{d_{0}}$, and that we proved in Section \ref{sec:compositions} that,
in both transformations, if the $\overrightarrow{\Lambda}$ (resp. $\Lambda'$) are $f$-\emph{checkable}, then so is $\bigoplus_{\overrightarrow{d_{0}}} 
\overrightarrow{\Lambda}$ (resp. $\Lambda \odot \Lambda'$).


\end{document}